\documentclass[
    aps,
    pra,
    twocolumn,
    amsmath,
    amssymb,
    nofootinbib,
    superscriptaddress]{revtex4-1}

\usepackage{physics}

\newcommand{\set}[1]{\left\lbrace #1\right\rbrace}

\usepackage[colorlinks]{hyperref} 
\usepackage[pdftex]{graphicx}
\usepackage{bbm}
\usepackage{amsthm}
\usepackage{cleveref}
\theoremstyle{definition}

\usepackage[usenames,dvipsnames]{xcolor}
\usepackage{framed}

\newtheorem{proposition}{Proposition}
\newtheorem{theorem}{Theorem}

\newtheorem{example}{Example}

\usepackage[normalem]{ulem} 
\usepackage{soul}

\hypersetup{
    colorlinks=true, 
    linktoc=all,     
    linkcolor=blue,  
    citecolor=blue,
    filecolor=blue,
    urlcolor=blue
}

\DeclareMathOperator*{\argmin}{argmin}
\DeclareMathOperator*{\argmax}{argmax}

\newcommand{\PhysMQ}{\affiliation{%
    Department of Physics and Astronomy,
    Macquarie University,
    Sydney NSW, Australia.}}
    
\newcommand{\PhysSU}{\affiliation{%
    Department of Physics,
    Stockholm University,
    Stockholm, Sweden.}}
    
\newcommand{\PhysUQ}{\affiliation{%
    School of Mathematics and Physics,
    University of Queensland,
    St. Lucia QLD, Australia.}}

\newcommand{\EQuSMQ}{\affiliation{%
    Centre for Engineered Quantum Systems,
    Macquarie University,
    Sydney NSW, Australia.}}

\newcommand{\EQuSUQ}{\affiliation{%
    Centre for Engineered Quantum Systems,
    University of Queensland,
    St. Lucia QLD, Australia.}}

\begin{document}

\title{Quantum State Discrimination as Bayesian Experimental Design}

\author{Thomas Guff}\email[]{thomas.guff@students.mq.edu.au} \PhysSU \PhysMQ \EQuSMQ
\author{Yuval R.\ Sanders} \PhysMQ  \EQuSMQ
\author{Nathan A.\ McMahon} \PhysUQ \EQuSUQ
\author{Alexei Gilchrist} \PhysMQ \EQuSMQ

\date{\today}

\begin{abstract}
We show that quantum state discrimination sits neatly in the framework of Bayesian experimental design. 
In this setting, the two main branches of quantum state discrimination (minimal error and maximal confidence) simply correspond to two different utility functions.
This view allows straightforward extensions and mixing of different discrimination tasks by examining the utility functions, and to describe multi-objective discrimination tasks. In addition, the probability of success and the total confidence are resource monotones quantifying the usefulness of measurements.
We give general conditions under which utility functions lead to resource monotones in the resource theory of quantum measurement. 
\end{abstract}

\maketitle

\section{Introduction}\label{sec:intro}

Quantum state discrimination~\cite{chefles_review_2000, barnett_review_2009,
bae_review_2015} is a foundational task in quantum information theory in which
an agent tries to distinguish between a collection of quantum states. If the states are
not pairwise orthogonal, then no measurement will be able to perfectly distinguish them.
Usually, one then seeks to find the \emph{optimal} measurement to discriminate between a given set of quantum states. 

Quantum state discrimination is emerging as significant in the field of resource theories as well.
Resource theories order states of a system in terms of their resourcefulness with respect to a particular task.
Discrimination tasks have shown to be useful in characterising resources in general resource theories \cite{Takagi2019}.
The degree of resourcefulness is quantified by real-valued functions known as resource monotones, which can be generated from state discrimination tasks \cite{Takagi2019,Skrzypczyk2019,Skrzypczyk2019a,Guff2019}.
In this context, quantum state discrimination is not necessarily used to find the `best' measurement, but to quantify the usefulness of any  (even a suboptimal) measurement.

Quantum state discrimination is typically thought to have two primary branches.
In the first, known as \emph{minimal error state discrimination}, the agent
attempts to report the correct state while minimising the probability of an
error or equivalently, maximising the overall success. In the second, known as
\emph{maximal confidence state discrimination}, the agent tries to maximise the
certainty (or confidence) in each of their measurement outcome, granting the additional ability to discard any measurement result.
Both these
cases involve the agent trying to best distinguish between quantum states, but they
differ in their respective notions of what constitutes `best'.

In this paper, we reframe both kinds of quantum state discrimination as a Bayesian experimental design problem.
This is a very general framework which is designed to select the experiment which maximises a chosen utility function. This utility function encodes what an agent most values in the experiment.
Quantum state discrimination naturally fits within this framework, where in place the best experiment, the agent searches for the best quantum measurement to discriminate between quantum states.
We show that both minimal error and maximal confidence state discrimination are Bayesian experimental design tasks, distinguished only by the choice of utility functions. 

This allows us to generalise quantum state discrimination tasks to other utility functions.
We provide a natural example of a utility function that leads to the mutual information between the measurement outcomes and the transmitted states as the function the agent wants to optimise.
The mutual information is a common measure of the average information extracted from a quantum system by a measurement device.

It also allows the description of multi-objective utility functions. These are useful when we may wish to optimise an experiment in multiple stages, which may be required to identify a unique optimal measurement.
We look at a well-known example of this for the case when there is a heavy bias towards one particular state. 
In this case a single stage of optimisation assigns the same utility value to all measurements.

The framework of Bayesian experimental design makes explicit the decision an agent will make after receiving the measurement result but before concluding which state they received.
In typical treatments of quantum state discrimination, this is usually subsumed into the measurement itself. 
We keep these notions separate, as it allows us to explore a wider range of POVMs and utility functions.

The probability of success, which arises when studying minimal error state discrimination, is a resource monotone in the resource theory of measurement \cite{Guff2019}. 
In fact, the probability of success over all possible discrimination alphabets forms a (uncountable) complete family of resource monotones \cite{Skrzypczyk2019,Guff2019}. 
This means that discrimination tasks can completely characterise the resource theory of quantum measurements.
Maximal confidence state discrimination also leads to a resource monotone, which we call the total confidence.
The question then arises as to which utility functions will give rise to resource monotones.
We provide three simple conditions on the utility function, which if satisfied, lead to resource monotones for quantum measurements.
When these conditions are satisfied there exists an optimal measurement with only rank-1 POVM elements. 
Thus in the future, it might be possible to use a broader range of utility functions to create resource monotones to better characterise the resource theory of quantum measurements.

\section{Quantum State discrimination}
\label{sec:formalism}

We first present the task of quantum state discrimination as it is commonly found in the literature \cite{chefles_review_2000, barnett_review_2009,
bae_review_2015}. 
In particular we expound the two primary branches: minimal error and maximal confidence state discrimination.
 As it will become necessary in the framework of Bayssian experimental design, we make explicit the decision strategy involved in selecting a quantum state; although in standard treatments it is not necessary to consider the decision explicitly.

The task of quantum state discrimination can be framed as a communication protocol between two parties.
Suppose one party, Alice, wishes to send a classical message to another, Bob.
Alice will choose one of $n$ distinct messages (labelled $\nu \in\mathcal{S}_{\nu}=  \set{1, \ldots, n}$)
and encode that message in an alphabet of quantum states, all of the same (finite) dimension, so that message $\nu$ is
encoded as density operator
$\rho_\nu$. Alice will then send $\rho_\nu$ to Bob using a noiseless quantum
communication channel. Bob knows that Alice intends to send message $\nu$ with
probability $q_\nu>0$ ($\sum_{\nu=1}^{n} q_\nu = 1$). Bob's task is to determine Alice's
message once he has received her state.

In order to infer Alice's message, Bob measures the received state.
His choice of measurement $\alpha$ can be represented with a positive operator-valued
measure (POVM). That is, for every measurement outcome $y\in  \mathcal{S}_{y} = \set{1,\dots, m}$, there is a positive semi-definite linear operator $A_{y}$, such that 
\begin{equation}
   \forall y:\ A_y \geq 0, \quad
    \sum_{y=1}^{m} A_y = \mathbbm{1}.
\end{equation}
Bob wants the best measurement to discriminate between the possible states Alice could send. Determining the best POVM is the primary task of quantum state discrimination. There are two main paradigms of quantum state discrimination corresponding to different notions of what Bob considers the `best' POVM.
Specifically, these are when Bob values making as few mistakes as possible, or having the largest possible confidence in each measurement outcome.

\subsubsection{Minimal Error State Discrimination}\label{sec:minerr}

When the aim is to make as few mistakes as possible when reading Alice's messages, Bob will choose the POVM that minimises the average probability of error, or equivalently, the POVM that maximises the average probability of success.
Let $N$ represent the possible states sent by Alice in one round of communication. Further, let $\alpha$ represent the measurement that Bob uses on the state Alice sent. We use the notation $N_{\nu}$ to represent the proposition that the state associated to message $\nu$ was sent by Alice. Similarly we use the notation $\alpha_{j}$ to represent the proposition that Bob's measurement $\alpha$ obtained the result $j$, from which he will conclude that Alice sent message $j$. 
With this notation, the probability of success is
      \begin{equation}
          \sum_{\nu=1}^{n} p\left(\alpha_{\nu}|N_{\nu}\right)p\left(N_{\nu}\right) = \sum_{\nu=1}^{n} q_\nu \tr \left( A_\nu \rho_\nu \right). 
          \label{eq:avesuccess}
      \end{equation}
Explicitly, this function is the average probability that Bob's measurement will display the correct outcome.
Finding the POVM which maximises the average probability of success is  \emph{minimal error quantum state discrimination}. Analytic solutions for minimal error discrimination are known only for a few
special cases such as $n=2$~\cite{helstrom_book_1976} or when the alphabet
exhibits specific symmetries \cite{Chou_mirror_2004, Andersson_mirror_2002}.
Necessary and sufficient conditions for the optimal POVM have been known
for some time \citep{helstrom_book_1976, barnett_conditions_2009}.

Focusing on the measurement, Eq.~\eqref{eq:avesuccess} is also \emph{measure of
the quality} of a particular POVM in the context of this particular discrimination task. 
That is, 'better' measurements are those which have a higher probability of success within this particular state discrimination task.
Here the `best' measurement has been determined to be the one that maximises the probability of success.

When defining the probability of success \eqref{eq:avesuccess}, we assumed that the number of measurement outcomes was the same as the number of states in the alphabet being sent to Bob by Alice.
This assumption arises from the implicit presumption that when Bob receives measurement outcome $j$, he decides that Alice sent message $j$.
To generalise this to POVMs with arbitrary numbers of
measurement outcomes, we need to make explicit the \emph{decision} that Bob will make when he receives a measurement outcome.

We make this explicit by defining a \emph{decision strategy} $\Gamma:\mathcal{S}_{y}\rightarrow \mathcal{S}_{\nu}$ as any function that maps measurement outcomes to the alphabet of Alice's possible  messages.
Let $D_{j}$ be the proposition that Bob's decision was that Alice sent $\rho_{j}$ and thereby wishes to communicate message $j$. Then the probability of success is the average probability that Bob's \emph{decision} coincides with the state that was sent to him,
\begin{equation}
\sum_{\nu=1}^{n} p\left(D_{\nu}|N_{\nu}\right)p\left(N_{\nu}\right). \label{eq:decsucc}
\end{equation}
The decision strategy $\Gamma$ may map multiple measurement outcomes from $\mathcal{S}_{y}$, to the same decision, $\nu$.
However we can combine the POVM elements corresponding to each decision without reducing the probability of success.
To see this note that if Bob makes decision $\nu$, then it must have been that his measurement device displayed a result which would, given the decision strategy $\Gamma$, have caused that decision. 
That is to say, the probability that Bob receives outcome $y$ and makes decision $\nu$ is
\begin{equation}
p\left(D_{\nu},\alpha_{y}|N_{\nu}\right) = p\left(\alpha_{y}|N_{\nu}\right)\delta_{\Gamma\left(y\right),\nu}
\end{equation}
With this, we can expand \eqref{eq:decsucc} over the $m$ measurement outcomes (where $m$ may not equal $n$),
\begin{align}
&\sum_{\nu=1}^{n} p\left(D_{\nu}|N_{\nu}\right)p\left(N_{\nu}\right) \nonumber \\
&=\sum_{\nu=1}^{n} \sum_{y=1}^{m} \;p\left(D_{\nu},\alpha_{y}|N_{\nu}\right)p\left(N_{\nu}\right) \nonumber \\
&= \sum_{\nu=1}^{n} \sum_{y=1}^{m}  p\left(\alpha_{y}|N_{\nu}\right)p\left(N_{\nu}\right) \delta_{\Gamma\left(y\right),\nu} \nonumber \\
&=\sum_{\nu=1}^{n} \; \sum_{y\in\Gamma^{-1}(\nu)} p\left(\alpha_{y}|N_{\nu}\right)p\left(N_{\nu}\right),\label{eq:multisucc}
\end{align}
where $\Gamma^{-1}(\nu)$ is the set of all measurement outcomes that leads Bob to decide $\nu$ given strategy $\Gamma$.
Since $p\left(\alpha_{y}|N_{\nu}\right)=\tr(A_{y}\rho_{\nu})$ is linear in the
POVM element, we can define a new measurement $\beta$ with POVM elements $B_{\nu} = \sum_{y\in\Gamma^{-1}(\nu)} A_{y}$, which will have the same probability of success as \eqref{eq:multisucc} (when using the corresponding decision strategy). Thus within the confines of minimal error state discrimination it is not necessary to consider measurement devices with more than $n$ outcomes; and this is seen in the literature. 
However to interpret it in the framework of Bayesian experimental design, we need to separate the measurement outcome from the decision strategy.

Even though, Bob may not have the optimal measurement device, he will still want to use the best decision strategy given his measurement device.
That is, his probability of success should be maximised over all decision strategies
\begin{equation}
\max_{\Gamma} \sum_{\nu=1}^{n} p\left(D_{\nu}|N_{\nu}\right)p\left(N_{\nu}\right).
\end{equation}
It is simple to show that the best decision strategy, in the case of minimal error state discrimination, is to always choose the state with the largest posterior probability after the measurement. To demonstrate this, let us assume a particular decision strategy $\Gamma$ and expand \eqref{eq:decsucc} in terms of the measurement outcomes
\begin{equation}
\sum_{y=1}^{m} \sum_{\nu=1}^{n} p\left(D_{\nu}|\alpha_{y},N_{\nu}\right)p\left(\alpha_{y}|N_{\nu}\right)p\left(\rho_{\nu}\right).\label{eq:expandsucc}
\end{equation}
Since the decision depends solely on the measurement outcome, we have $p\left(D_{\nu}|\alpha_{y},N_{\nu}\right)=p\left(D_{\nu}|\alpha_{y}\right)$, and $p\left(D_{\nu}|\alpha_{y}\right)$ represents the decision strategy. Since the decision strategy is a deterministic function of the measurement outcomes, the probabilities $p(D_{\nu}| \alpha_{y})$ are either $1$ or $0$. 

If we define $\rho_{y^{\star}}$ as the state with the largest posterior probability given measurement outcome $y$,
\begin{equation}
y^\star := \argmax_{\nu\in\set{1,\dots,n}} p\left(N_{\nu}|\alpha_{y}\right),
\end{equation}
then, applying Bayes' rule, we have
\begin{align}
\sum_{y=1}^{m} \sum_{\nu=1}^{n} &\;p\left(D_{\nu}|\alpha_{y}\right) p\left(N_{\nu}|\alpha_{y}\right)p\left(\alpha_{y}\right) \nonumber \\
&\leq \sum_{y=1}^{m} \sum_{\nu=1}^{n} p\left(D_{\nu}|\alpha_{y}\right) p\left(N_{y^\star}|\alpha_{y}\right)p\left(\alpha_{y}\right) \nonumber \\
&= \sum_{y=1}^{m} p\left(N_{y^{\star}}|\alpha_{y}\right)p\left(\alpha_{y}\right). \label{eq:bestme}
\end{align}
This inequality can clearly can be saturated since the last line corresponds to the probability of success when the decision strategy was to choose the state with the largest posterior probability; that is, $p\left(D_{\nu}|\alpha_{y}\right) = \delta_{\nu,y^{\star}}$.

\subsubsection{Maximal Confidence State Discrimination}

If the priority is having the maximal possible assurance that each measurement outcome is accurate, a different POVM will be considered the `best' or optimal measurement.

If the states in Alice's alphabet are mutually orthogonal then there exists a measurement that can discriminate between them with certainty and without error. 
However, for certain alphabets which aren't necessarily mutually orthogonal, Bob can choose a POVM such that each element is orthogonal to all but one of the possible states. 
In this case there is a subspace in the support of each state $\rho_{\nu}$ that is not in the support of any other state Alice could send.
Hence the states $\rho_{1},\dots,\rho_{n}$ will be linearly independent.
If this subspace condition holds then one can construct a POVM such that  $p\left(N_{\nu}|\alpha_{y}\right) = \delta_{\nu,y}$ for all $y$. That is, 
Bob is certain about which state he was sent irrespective of which measurement outcome, $y\in \{1,\dots, m\}$, he received.
This is known as \emph{unambiguous state discrimination}
\cite{Peres_unambiguous_1988, Ivanovic_unanbiguous_1987, dieks_overlap_1988,
Jaeger_unambiguous_1995, Eldar_sdp_2004}. 
For an example, see example \ref{ex:unambig} below.

For non-orthogonal alphabets, unambiguous state discrimination can only be done at the cost of adding an `inconclusive' outcome $A_{0}$ to ensure the POVM sums to unity. 
By convention when Bob obtains the inconclusive outcome, he does not make a decision as to which state was sent.
Unambiguous state discrimination is possible for $n$ pure states if and only if they are linearly independent~\cite{chefles_unambiguous_1998}. 
Solutions for unambiguous state discrimination exist for up to three linearly independent pure states \cite{Peres_three_1998}, and $n$ linearly independent symmetric states \cite{chefles_optimum_1998}.

The generalisation to all alphabets is known as \emph{maximal confidence measurements} \cite{croke_maximum_2006}, where Bob may no longer have certainty in each measurement outcome, but will be as certain as he possibly can be; that is he will have the most confidence. Here Bob constructs the following measurement $\gamma$.
For each message $\nu = 1, \dots, n$, he chooses POVM element $A_\nu$ to associate to the measurement outcome $\nu$, such that it maximises $p( N_\nu | \gamma_\nu )$,
\begin{equation}
A_{\nu} = \argmax_{0< A \leq \mathbbm{1}}\; p\left(N_{\nu}|\gamma_{\nu}\right) = \argmax_{0< A \leq \mathbbm{1}} \frac{q_{\nu}\tr(\rho_{\nu}A)}{\tr\left(A\left(\sum_{\nu=1}^{n}q_{\nu}\rho_{\nu}\right)\right)}. \label{eq:maxconf}
\end{equation}
Here we are maximising over all POVM elements $A$ for which $p\left(\gamma_{\nu}\right)=\tr\left(A\left(\sum_{\nu=1}^{n}q_{\nu}\rho_{\nu}\right)\right)\ne 0$. These are the positive semi-definite operators that are not orthogonal to every state in Alice's alphabet.
He then defines the inconclusive outcome $A_{0} := \mathbbm{1} - \sum_{\nu=1}^{n} A_\nu$. Since the trace function is linear, $p\left(N_{\nu}|\gamma_\nu\right)$ is unchanged when $A_{\nu}$ is rescaled by a constant. Hence the elements $A_{\nu}$ can be rescaled to ensure that $A_{0}$ is positive semi-definite.

This ability to rescale elements implies that there is a (continuous) family of POVMs that maximises \eqref{eq:aggconf}.  
To select a unique member of this family, Bob typically performs another optimisation: choosing, for example, the POVM from this family that minimises the probability of obtaining the inconclusive outcome, $p(\gamma_{0})$.
This is an example of multi-objective optimisation, in which a secondary optimisation occurs subject to the constraint of a previous optimisation. 
We will revisit this in section~\ref{sec:multiobj}.

Unlike the best minimal error quantum measurement, which maximises the probability of success \eqref{eq:avesuccess}, the maximal confidence measurement is constructed one element at a time. Thus is it unclear how to measure the quality, or `confidence', of a suboptimal measurement $\alpha$ with outcomes $y=0,1,\dots,n$. We have not found such a measure in the literature, but we believe it to be intuitive and reasonable to use the sum of the confidences of each outcome
\begin{equation}
\sum_{\nu=1}^{n} p\left(N_{\nu}|\alpha_{\nu}\right),\label{eq:aggconf}
\end{equation}
where if $p\left(\alpha_{\nu}\right)=0$ then we define $p\left(N_{\nu}|\alpha_{\nu}\right)=0$. We call this measure the \emph{total confidence}.
Clearly, the POVM constructed in \eqref{eq:maxconf} independently optimise each term in this sum, thereby maximising the total sum.

Once again, we would like to consider more general POVMs by separating the measurement outcome from the decision Bob makes after his measurement. 
In this case we must include an `inconclusive' decision, so now the decision strategy $\Gamma$ maps into the set $\mathcal{S}_\nu = \set{0,\dots,n}$. 
The quality of our measurement device is the probability that Alice sent $\rho_{\nu}$ given Bob decided on $\rho_{\nu}$. Using the notation from sec.~\ref{sec:minerr}, this is
\begin{align}\label{eq:confdec}
\sum_{\nu=1}^{n} p\left(N_{\nu}|D_{\nu}\right) &= \sum_{\nu=1}^{n} \frac{p\left(D_{\nu}|N_{\nu}\right)p\left(N_{\nu}\right)}{p\left(D_{\nu}\right)} \nonumber \\ 
&= \sum_{\nu=1}^{n} \frac{\sum_{y\in\Gamma^{-1}\left(\nu\right)}p(\alpha_{y}|N_{\nu}) p\left(N_{\nu}\right)}{\sum_{y\in\Gamma^{-1}\left(\nu\right)}p( \alpha_{y})} \nonumber \\ 
&= \sum_{\nu=1}^{n} p(N_{\nu}|\textstyle{\sum_{y\in\Gamma^{-1}\left(\nu\right)}} \alpha_{y}),
\end{align}
where the summation of propositions denotes logical disjunction: $\alpha_{1}+\alpha_{2}+\dots$ means $\alpha_{1}$ \emph{or} $\alpha_{2}$ \emph{or} $\dots$.

Once again, Bob wishes to use the best decision strategy to maximise the total confidence of his measurement,
\begin{equation}
\max_{\Gamma} \sum_{\nu=1}^{n} p\left(N_{\nu}|\alpha_{\nu}\right). \label{eq:totalconf}
\end{equation}
Unlike minimal error state discrimination, Bob's best strategy here is not to choose the state with the largest posterior probability. 
The best decision strategy will be to assign a single measurement outcome to each decision and to make no decision on all remaining measurement outcomes. 
This is because making the same (conclusive) decision on two measurement outcomes can never do better than making that decision on only one of them. To show this consider a decision strategy $\Gamma$ such that $\Gamma\left(i\right) = \Gamma\left(j\right) = \nu \ne 0$, where $i\ne j$. Then if we examine we  $p\left(N_{\nu}|D_{\nu}\right)$, we can rewrite it as
\begin{align}
p\left(N_{\nu}|D_{\nu}\right) &=  p\left(N_{\nu}|\alpha_{i}+\alpha_{j}\right) = \frac{p\left(\alpha_{i}+\alpha_{j}|N_{\nu}\right)p\left(N_{\nu}\right)}{p\left(\alpha_{i}+\alpha_{j}\right)} \nonumber \\
&=\frac{p\left(N_{\nu}|\alpha_{i}\right)p\left(\alpha_{i}\right)+p\left(N_{\nu}|\alpha_{j}\right)p\left(\alpha_{j}\right)}{p\left(\alpha_{i}+\alpha_{j}\right)}.
\end{align}
Since this is an weighted average of $p\left(N_{\nu}|\alpha_{i}\right)$ and $p\left(N_{\nu}|\alpha_{j}\right)$ we have
\begin{equation}
p\left(N_{\nu}|D_{\nu}\right) \leq \max\set{p\left(N_{\nu}|\alpha_{i}\right),p\left(N_{\nu}|\alpha_{j}\right)}.\label{eq:mcbestdec}
\end{equation}
In other words, a better strategy would have been to have made decision $\nu$ on the strategy on the measurement outcome with the larger confidence, and to make no decision on the other. Of course if $p\left(N_{\nu}|\alpha_{i}\right)=p\left(N_{\nu}|\alpha_{j}\right)$, then
\begin{equation}
p(N_{\nu}|D_{\nu})=p\left(N_{\nu}|\alpha_{i}\right)=p\left(N_{\nu}|\alpha_{j}\right)=p\left(N_{\nu}|\alpha_{i}+\alpha_{j}\right)
\end{equation}
and the confidence in decision $\nu$ is unaffected. 

Including Bob's decision allows us to consider POVMs with more elements than the number of messages in Alice's alphabet, now we can re-interpret quantum state discrimination in the language of Bayesian experimental design. However, just like with minimal state discrimination, this does not allow Bob to design a more optimal POVM than the standard approach allows. 

\section{Bayesian experimental design}
\label{sec:BED}

Here we describe the general formalism of Bayesian experimental design, before showing that quantum state discrimination straightforwardly fits into this framework. The two different paradigms mentioned in the previous section will simply correspond to different utility functions

\subsection{General Formalism}
\label{ssec:formalism}

Bayesian experimental design is a framework for deriving the best experiment for a particular task, where the notion of `best' is encoded into a utility function $U$.

An intuitive introduction to Bayesian experimental design is through Bayesian
decision theory \cite{ProbJaynes2003}, which involves both inference and optimisation.
We begin with a problem in which we have an unknown variable $N$ which can take on any value $\nu$ from the set $\mathcal{S}_{\nu}$, of which we would usually like to know more.
We will assume that all sets are finite unless stated otherwise.
There will be some background information $\epsilon\in\mathcal{S}_{\epsilon}$ about the problem (for example the particular experiment being performed), and there will some relevant data $y\in\mathcal{S}_{y}$ about $N$ such as experimental or measurement results.
We use $\epsilon_{y}$ to represent the proposition that the experiment $\epsilon$ generated data $y$.
Our knowledge of the variable $N$ is then captured by the conditional probability distribution $p(N_{\nu}|\epsilon_{y})$; where we have used the notation $N_{\nu}$ to stand for the proposition that $N=\nu$.

The experimenter will then make a decision, represented by variable $D$. The proposition that the particular decision $d$ from a set of possible decisions $\mathcal{S}_{d}$ is chosen is represented by $D_{d}$. 
Usually this is a decision as to the value of $N$. What we value in this experiment, which depends on all the parameters for generality, is encoded in a utility function
\begin{equation}
U:\mathcal{S}_{\epsilon}\times\mathcal{S}_{y}\times\mathcal{S}_{d}\times\mathcal{S}_{\nu}\rightarrow \mathbb{R}^{n},
\end{equation}
where $\mathbb{R}^{n}$ is considered as a vector space, and totally ordered under the \emph{dictionary order} (also known as the \emph{lexicographic order}). For example, in $\mathbb{R}^{2}$, $\left(a,b\right) > \left(c,d\right)$ if $a > c$ or if $a = c$ and $b > d$.

The objective in Bayesian decision theory is then to pick the optimal decision $d^\star$.  
The optimisation process should clearly be weighted by $p\left(N_{\nu}|\epsilon_{y}\right)$, which the best knowledge we have of $\nu$. So for a given datum $y$, the optimal decision should have the largest average utility over $P(N_{\nu}|\epsilon_{y})$. That is,
\begin{align}
    d^\star &= \argmax_{d\in\mathcal{S}_d} 
            \sum_{\nu\in\mathcal{S}_{\nu}} U(\epsilon,y,d,\nu)\,p(N_{\nu}|\epsilon_{y}).
            \label{eq:bestdec}
\end{align}
So the optimal decision yields a utility of
\begin{align}
U(\epsilon, y) = \max_{d\in\mathcal{S}_{d}}
  \sum_{\nu\in\mathcal{S}_{\nu}} U(\epsilon,y,d,\nu)\,p(N_{\nu}|\epsilon_{y}).
\end{align}

The extension to experimental design is then straightforward: we desire the best experiment $\epsilon^\star$ from a set of possible experiments $\mathcal{S}_{\epsilon}$, averaged over all possible data $y\in \mathcal{S}_{y}$ we could have received,
\begin{align}
    \epsilon^\star &= \argmax_{\epsilon\in\mathcal{S}_\epsilon} \sum_{y\in\mathcal{S}_y} p(\epsilon_{y})\, \max_{d\in\mathcal{S}_d} 
    \sum_{\nu\in\mathcal{S}_\nu} U(\epsilon,y,d,\nu) \,p(N_{\nu}|\epsilon_{y}) \nonumber \\
    &=\argmax_{\epsilon\in\mathcal{S}_\epsilon} \sum_{y\in\mathcal{S}_y} \max_{d\in\mathcal{S}_d} \sum_{\nu\in\mathcal{S}_\nu} p(N_{\nu},\epsilon_{y}) \,U(\epsilon,y,d,\nu).\label{eq:oldbed}
\end{align} 
The last equation is in the form presented by the influential review
\cite{95chaloner273}, and is notable for explicitly including the intermediate decision optimisation. 

Instead of simply maximising over the possible decisions as an intermediate step, we wish to generalise this to optimise over all decision strategies more generally. A \emph{decision strategy} $\Gamma:\mathcal{S}_{y}\rightarrow \mathcal{S}_d$ is a function from the set of possible data that could have been received $\mathcal{S}_{y}$ to the set of possible decisions $\mathcal{S}_{d}$. The set of all possible decision strategies $\mathcal{S}_{\Gamma}=\left.\mathcal{S}_{d}\right.^{\mathcal{S}_{y}}$ is thus the set of all functions from $\mathcal{S}_{y}$ to $\mathcal{S}_{d}$.
The utility function $U$ is now
\begin{equation}
U:\mathcal{S}_{\epsilon}\times\mathcal{S}_{y}\times\mathcal{S}_{\Gamma}\times\mathcal{S}_{\nu}\rightarrow \mathbb{R}^{n}
\end{equation}
The expression for the best experiment, is now generalised to
\begin{equation}
\epsilon^\star = \argmax_{\epsilon\in\mathcal{S}_\epsilon} \,\max_{\Gamma\in\mathcal{S}_{\Gamma}}\sum_{y\in\mathcal{S}_y} \sum_{\nu\in\mathcal{S}_\nu} p(N_{\nu},\epsilon_{y})\, U(\epsilon,y,\Gamma,\nu). \label{eq:bed}
\end{equation}

This generalised expression returns to \eqref{eq:oldbed} in the case in which the utility function depends not on the entire decision strategy $\Gamma$, but only on the particular decision made $d=\Gamma\left(y\right)$; in which case $U(\epsilon,y,\Gamma,\nu) \equiv U(\epsilon,y,d,\nu)$.
This occurs when the best decision strategy is to simply make the optimal decision for each measurement outcome individually.
Here the best decision strategy is to choose $\Gamma(y)$ such that
\begin{equation}
\Gamma\left(y\right) = \argmax_{d\in S_{d}} \sum_{\nu\in S_{\nu}} p(N_{\nu},\epsilon_{y})\, U(\epsilon,y,d,\nu).
\end{equation}
This implies that the optimal measurement satisfies
\begin{equation}
\argmax_{\epsilon\in\mathcal{S}_\epsilon} \sum_{y\in\mathcal{S}_y} \max_{d\in\mathcal{S}_{d}} \sum_{\nu\in\mathcal{S}_\nu} p(N_{\nu},\epsilon_{y})\, U(\epsilon,y,d,\nu),
\end{equation}
which is the less general form we introduced earlier \eqref{eq:oldbed}.
We saw earlier that the best decision strategy was to choose the state largest posterior probability individually for each measurement outcome. We will see that minimal error can be derived from a utility function which depends only on the particular decision made. In contrast, we saw that the best decision strategy for the maximal confidence measurements is not to choose the state for which Bob has the largest confidence after each measurement outcome. This implies that maximal confidence cannot be derived from such a utility function.


\subsection{Application to State Discrimination}%
\label{ssec:specialisation_to_state_discrimination}

This framework can be specialised to both minimal error and maximal confidence quantum state discrimination. In both cases, rather than searching for the optimal experiment we are searching for the optimal quantum measurement.
Thus the set of experiments $\mathcal{S}_\epsilon$ is the (infinite) set of all POVMs. More precisely, let $\mathcal{S}^{(m)}_{\epsilon}$ be the set of all $m$-tuples of positive semi-definite linear operators (all defined on the same finite-dimensional Hilbert space)
\begin{equation}
(A_{1}, \dots, A_{m})
\end{equation}
such that $\sum_{i=1}^{m} A_{i} = \mathbbm{1}$. Then the set of all POVMs $\mathcal{S}_{\epsilon}$ is
\begin{equation}
\mathcal{S}_{\epsilon} = \bigcup_{m=1}^{\infty} \mathcal{S}^{(m)}_{\epsilon}.
\end{equation}
We could instead choose to optimise over consider a restricted subset of POVMs, this restriction may have physical motivations due to an experimenter's laboratory set-up or technological limitations.
However for the purposes of this paper, we will allow the experimenter access to all POVMs, as is typical in the literature.

The set of possible output data is the possible outcomes of the particular POVM $\epsilon$.
If $m$ is the number of elements in the particular POVM $\epsilon$, then $\mathcal{S}_{y}=\set{1,\dots,m}$.

The unknown parameter is the message sent by Alice, so $\mathcal{S}_{\nu} = \set{1,\dots,n}$, which she will encode into a quantum state.
In quantum state discrimination the set of decisions is $\mathcal{S}_{d}=\mathcal{S}_{\nu}\cup\set{0}$; the possible states that Alice can send along with an `inconclusive' decision, where Bob makes no decision as to which state Alice sent.
So Bob will choose the measurement which maximises 
\begin{equation}
    \epsilon^{\star} = \underset{\epsilon\in \mathcal{S}_{\epsilon}}{\argmax}\, \max_{\Gamma\in \mathcal{S}_{\Gamma}} \sum_{y=1}^{m} \sum_{\nu=1}^{n} p(N_{\nu},\epsilon_{y})\, U(\epsilon,y,\Gamma,\nu). \label{eq:qsdbed}
\end{equation}
Initially, we will only consider a single-objective optimisation: that is, $U(\epsilon,y,\Gamma,\nu) \in \mathbb{R}$.
Since the set $\mathcal{S}_{\epsilon}$ is infinite, it is possible that no POVM maximises the utility function. Usually the averaged utility function
\begin{equation}
U(\epsilon) = \max_{\Gamma\in \mathcal{S}_{\Gamma}} \sum_{y=1}^{m} \sum_{\nu=1}^{n} p(N_{\nu},\epsilon_{y})\, U(\epsilon,y,\Gamma,\nu)\label{eq:Uep}
\end{equation}
will be bounded, so it cannot take on arbitrarily large values. However, this does not guarantee that there exists a POVM $\epsilon^{\star}$. Nevertheless, the function $U(\epsilon)$ can be interpreted as quantifying how `good' measurement $\epsilon$ is at this particular task. In section~\ref{sec:conditions} we will provide conditions under which $U(\epsilon)$ can be considered a resource monotone in the resource theory of quantum measurements.

\subsubsection{Minimal Error State Discrimination}

In minimal error state discrimination, Bob values minimal error in his decision. That is, the utility function is simply
\begin{equation}
    U(\epsilon,y,\Gamma,\nu) = \delta_{\Gamma\left(y\right),\nu}\, ,
\end{equation}
where Bob's utility is zero if his decision does not correspond to the state sent by Alice, and one if it does.
Note that $\Gamma(y)\in\set{0,\dots,n}$ but $\nu\in \set{1,\dots,n}$, so  $\delta_{0,\nu}=0$ for all $\nu$.

Since this utility function is only dependent upon the specific decision Bob makes after his measurement, we can use the form \eqref{eq:oldbed}.
For this utility function, the optimal POVM $e^\star$ satisfies
\begin{equation}
    \epsilon^{\star} = \argmax_{\epsilon\in \mathcal{S}_{\epsilon}} \sum_{y=1}^{n} \max_{d\in S_{d}}\, p\left(N_{d}|\epsilon_{y}\right)\, p\left(\epsilon_{y}\right).
\end{equation}
By inspection, this expression is equivalent to \eqref{eq:bestme}: Bob wants the POVM which maximises the average probability that the measurement correctly identifies the state sent by Alice when using the best decision strategy: namely choosing the state which maximises the posterior probability. 
Note that since the inconclusive outcome counts as an error, Bob will never choose the inconclusive decision in the context of minimal error quantum state discrimination.

\subsubsection{Maximal Confidence State Discrimination}

In maximal confidence state discrimination, Bob values the confidence in each decision he makes.
It can be derived it from the utility function
\begin{equation}\label{eq:MaxLikelyUtility}
    U\left(\epsilon,y,\Gamma,\nu\right) = \frac{\delta_{\Gamma\left(y\right),\nu}}{p\left(D_{\Gamma\left(y\right)}\right)},
\end{equation}
where Bob's utility is inversely proportional to the probability of making that decision, to compensate for how rarely that decision is made. 
Note that $p(D_{\Gamma(y)})$ is calculated under the assumption that Bob has chosen measurement $\epsilon$,
\begin{align}
p(D_{\Gamma(y)}) &= \sum_{j=1}^{m} p(D_{\Gamma(y)}|\epsilon_{j})p(\epsilon_{j}) \nonumber\\
&= \sum_{j\in \Gamma^{-1}(y)} p(\epsilon_{j}).
\end{align}

Although this utility function might appear to only depend on a single decision, it actually depends on the entire decision strategy, since $p\left(D_{\Gamma\left(y\right)}\right)$ is only calculable with the knowledge of the entire decision strategy $\Gamma$, for we need to know how many measurement outcomes map to the same decision $\Gamma\left(y\right)$.

From \eqref{eq:bed}, the optimal POVM will satisfy
\begin{align}
    \epsilon^{\star} &= \argmax_{\epsilon\in \mathcal{S}_\epsilon} \,\max_{\Gamma\in \mathcal{S}_{\Gamma}} \sum_{y=1}^{m} \sum_{\nu=1}^{n} \frac{p\left(N_{\nu},\epsilon_{y}\right)\delta_{\Gamma\left(y\right),\nu}}{p\left(D_{\Gamma\left(y\right)}\right)} \nonumber \\
    &= \argmax_{\epsilon\in \mathcal{S}_\epsilon} \,\max_{\Gamma\in \mathcal{S}_{\Gamma}}  \sum_{\nu=1}^{n} \sum_{y\in \Gamma^{-1}(\nu)} \frac{p\left(N_{\nu},\epsilon_{y}\right)}{p\left(D_{\nu}\right)} \nonumber \\
    &= \argmax_{\epsilon\in \mathcal{S}_\epsilon} \,\max_{\Gamma\in \mathcal{S}_{\Gamma}}  \sum_{\nu=1}^{n} p\left(N_{\nu}|D_{\nu}\right),\label{eq:derivemc}
\end{align}
where in the last equality we used the fact that $\sum_{y\in \Gamma^{-1}(\nu)} p\left(N_{\nu}, \epsilon_{y}\right) = p\left(N_{\nu},D_{\nu}\right)$.
This final expression \eqref{eq:derivemc} is that which we have for a maximal confidence measurement \eqref{eq:confdec}.

As mentioned in section~\ref{sec:formalism}, in general there will be a family of measurements which maximise the total confidence. One can then prefer a member of this family by requiring an additional optimisation; such as the member which minimises the probability of the inconclusive outcome. We return to this idea in section~\ref{sec:multiobj}. 


\section{Conditions for Resource Monotones}\label{sec:conditions}

We have seen that the probability of success \eqref{eq:bestme} and the total confidence \eqref{eq:totalconf} serve as measures of the quality of a POVM at a particular discrimination task. In this section we formalise that notion, by showing that these measures are actually resource monotones in the resource theory of quantum measurements. 
This implies that if you stochastically scramble any POVM it will perform worse, according to these measures, at discriminating quantum states for any alphabet Alice uses.
To show this we first describe general conditions under which the utility function $U(\epsilon,y,\Gamma,\nu)$, will produce a resource monotone $U(\epsilon)$ \eqref{eq:Uep}.

Let measurement $\alpha$ be represented by $\left(A_{1},\dots,A_{m}\right)$ and measurement $\beta$ be represented by $(B_{1},\dots,B_{m^{\prime}})$, then the resource theory of quantum measurements orders $\alpha \geq \beta$, if we can write for all $i=1,\dots,m^{\prime}$,
\begin{equation}
B_{i} = \sum_{j=1}^{m} p(\beta_{i}|\alpha_{j})A_{j}.
\end{equation}
That is, $\beta$ is a stochastic mixing of $\alpha$. This transformation of $\alpha$ into $\beta$ is sometimes known as post-processing. 
A \emph{resource monotone} in the resource theory of quantum measurements \cite{Guff2019}, is a function
\begin{equation}
f:\mathcal{S}_{\epsilon} \rightarrow \mathbb{R},
\end{equation}
such that if $\alpha \geq \beta$, then $f(\alpha) \geq f(\beta)$. So a resource monotone assigns a real number to each measurement, and it respects the ordering of the resource theory.
It is well known that the probability of success is a resource monotone for the resource theory of quantum measurements \cite{Skrzypczyk2019,Guff2019}. Indeed, the total confidence \eqref{eq:confdec} is a monotone as well.
This suggests that the utility $U(\epsilon)$ \eqref{eq:Uep} may generally be resource monotones in the case when $U(\epsilon)\in \mathbb{R}$. However the utility function $U(\epsilon,y,\Gamma,\nu)$ is so general (we have hitherto placed no restrictions on it) that $U(\epsilon)$ will not typically be a resource monotone.
For example, if $\alpha \geq \beta$, then a possible utility function could satisfy $U(\beta,y,\Gamma,\nu) = 1$ for all $y$, $\Gamma$ and $\nu$, and $U(\alpha,y,\Gamma,\nu) = 0$ for all $y$, $\Gamma$ and $\nu$; this would clearly not lead to a resource monotone.
In this section we detail three conditions on the utility function $U(\epsilon,y,\Gamma,\nu)$ which guarantee that $U(\epsilon)$ will become resource monotone.

\begin{itemize}
\item[\textbf{C1}] Let the measurement $\alpha$ be represented by the POVM $(A_{1},\dots,A_{m})$ with decision strategy $\Gamma$. Suppose that $A_{i},A_{j}$ are both non-zero and proportional to each other, and that $\Gamma(i) = \Gamma(j)$. Then $U(\alpha,i,\Gamma,\nu) = U(\alpha,j,\Gamma,\nu)$ for all $\nu=1,\dots,n$.
\item[\textbf{C2}] Let the measurement $\alpha$ be represented by the POVM $(A_{1},\dots,A_{m})$ . Suppose that $A_{i},A_{j}$ are non-zero and proportional to each other. Then the optimal decision strategy,
\begin{equation}
\Gamma^{*}=\argmax_{\Gamma\in  \mathcal{S}_{\Gamma}} \sum_{y=1}^{m}\sum_{\nu=1}^{n} p(N_{\nu},\alpha_{y})U(\alpha,y,\Gamma,\nu),
\end{equation}
assigns $\Gamma^{*}(i)=\Gamma^{*}(j)$.
\item[\textbf{C3}] Let the measurement $\alpha$ be represented by the POVM $(A_{1},\dots,A_{i},\dots,A_{j},\dots,A_{m})$ with decision strategy $\Gamma$. Suppose that $\Gamma(i) = \Gamma(j)$, and consider the $m-1$ outcome measurement $\alpha^{\prime}$ represented by POVM 
\begin{equation}
(A_{1},\dots,A_{i}+A_{j},\dots,A_{m}).
\end{equation}
Let us label the measurement outcomes for this new POVM $\set{1,\dots,i+j,\dots,m}$ and also consider induced decision strategy $\Gamma^{\prime}$ that assigns to every outcome in $\alpha^{\prime}$ the same decision as was assigned by $\Gamma$ to measurement outcomes of $\alpha$, in particular we define $\Gamma^{'}(i+j)=\Gamma(i)=\Gamma(j)$. Then for all $\nu = 1, \dots, n$,
\begin{align}
U&(\alpha^{\prime},i+j,\Gamma^{\prime},\nu) \nonumber \\
&\leq  \frac{p(\alpha_{i}|N_{\nu})U(\alpha,i,\Gamma,\nu)+p(\alpha_{j}|N_{\nu})U(\alpha,j,\Gamma,\nu)}{p(\alpha_{i}+\alpha_{j}|N_{\nu})},
\end{align}
with equality if $A_{i}$ and $A_{j}$ are proportional and non-zero; and
\begin{equation}
U(\alpha^{\prime},k,\Gamma^{\prime},\nu) = U(\alpha,k,\Gamma,\nu),
\end{equation}
for all $k\ne i,j$.
\end{itemize}

These three conditions characterise a utility function that values the ability of a quantum measurement to extract information from the quantum state. The first two conditions describe POVMs with non-zero elements that are proportional.
From the perspective of information extraction, measurement outcomes from proportional elements provide the same amount of information; therefore their utilities should be the same, and the best decision strategy should make the same decision for both outcome.
The third condition relates to merging two measurement outcomes, a transformation that can only degrade the ability of a measurement to extract information. Here we require that the utility is at best the average of the utility of the two prior measurement outcomes.

We now prove that if $U(\epsilon,y,\Gamma,\nu)$ satisfies these three conditions then $U(\epsilon)$ is a resource monotone.
\begin{theorem}
Let the measurement $\alpha$ be represented by the POVM $(A_{1},\dots,A_{m})$, and $\beta$ be represented by $(B_{1},\dots,B_{m^{\prime}})$, related to $\alpha$ by stochastic mixing
\begin{equation}
B_{i} = \sum_{j=1}^{m} p(\beta_{i}|\alpha_{j})A_{j},\label{eq:mix}
\end{equation}
for all $i=1,\dots,m^{\prime}$, and let a utility function $U(\epsilon,y,\Gamma,\nu)$ satisfy conditions $\textbf{C1}$, $\textbf{C2}$ and $\textbf{C3}$, then
$U(\alpha)\geq U(\beta)$, where $U(\epsilon)$ is defined as in \eqref{eq:Uep}.
\end{theorem}
\begin{proof}
Let $U(\epsilon,y,\Gamma,\nu)$ satisfy conditions \textbf{C1}, \textbf{C2} and \textbf{C3}.
Let $\Gamma^{*}_{\beta}$ be the optimal decision strategy for measurement $\beta$, that is
\begin{equation}
\Gamma^{*}_{\beta}=\argmax_{\Gamma\in  \mathcal{S}_{\Gamma}} \sum_{y=1}^{m^{\prime}}\sum_{\nu=1}^{n} p(N_{\nu},\beta_{y})U(\beta,y,\Gamma,\nu).
\end{equation}
Now let us create the $(m\times m^{\prime})$-outcome measurement $\beta^{\prime}$, whose POVM is created by using every term in all the sums represented by \eqref{eq:mix},
\begin{equation}
\left( p(\beta_{1}|\alpha_{1})A_{1},\dots, p(\beta_{m^{\prime}}|\alpha_{m})A_{m} \right).
\end{equation}
Let us label the measurement outcomes of this new POVM by the pair $(i,j)$, where $(i,j)$ corresponds to the POVM $p(\beta_{i}|\alpha_{j})A_{j}$.
We associate to this POVM the induced decision strategy $\Gamma^{*\prime}_{\beta}$ which makes the same decision as was made for the measurement $\beta$; that is $\Gamma_{\beta}^{*\prime}(i,j) = \Gamma_{\beta}^{*}(i)$.

Now it is clear that $\beta$ (with decision strategy $\Gamma_{\beta}^{*}$) can be derived from $\beta^{\prime}$ (with decision strategy $\Gamma_{\beta}^{*\prime}$) using a finite number of merging operations described in \textbf{C3}: $\beta$ is recovered by combining all outcomes with first index $i$. This can be done one at a time, each time applying the inequality described in \textbf{C3}. Finally, we have
\begin{align}
U(\beta) &= \sum_{y=1}^{m^{\prime}} \sum_{\nu=1}^{n} p(N_{\nu},\beta_{y})\, U(\beta,y,\Gamma^{*}_{\beta},\nu) \nonumber \\
&\leq \sum_{i=1}^{m^{\prime}}\sum_{j=1}^{m} \sum_{\nu=1}^{n} p(N_{\nu},\beta^{\prime}_{(i,j)})U(\beta^{\prime},(i,j),\Gamma^{*\prime}_{\beta},\nu).
\end{align}
Since the strategy $\Gamma_{\beta}^{*\prime}$ is possibly suboptimal, we know that
\begin{align}
U(\beta) &\leq \max_{\Gamma\in \mathcal{S}_{\Gamma}}\sum_{i=1}^{m^{\prime}}\sum_{j=1}^{m} \sum_{\nu=1}^{n} p(N_{\nu},\beta^{\prime}_{(i,j)})U(\beta^{\prime},(i,j),\Gamma,\nu), \nonumber \\
&= U(\beta^{\prime}).\label{eq:ub}
\end{align}
Now $\beta^{\prime}$ has many POVM elements which are proportional to each other; specifically the element corresponding to outcome $(i,j)$ is proportional to the element corresponding to $(i^{\prime},j)$ for all $i,i^{\prime},j$.
By condition \textbf{C2}, the optimal decision strategy $\Gamma^{*}_{\beta^{\prime}}$ will assign the same state to proportional, non-zero elements,
\begin{equation}
\Gamma^{*}_{\beta^{\prime}}(i,j) = \Gamma^{*}_{\beta^{\prime}}(i^{
\prime},j) \quad \text{for all } i,i^{\prime},j. \label{eq:bestdecstrat}
\end{equation}
It is irrelevant how we assign decisions or even utilities for measurement outcomes represented by POVM elements which are zero. This is because they will occur with zero probability, and will contribute nothing towards $U(\epsilon)$. Hence we can assume all elements are non-zero without loss of generality.

By condition \textbf{C1} the utility function is constant over non-zero proportional POVM elements which are assigned the same decision; that is we can write
\begin{equation}
U(\beta^{\prime},(i,j),\Gamma^{*}_{\beta^{\prime}},\nu) = U(\beta^{\prime},j,\Gamma^{*}_{\beta^{\prime}},\nu),
\end{equation}
for all $i=1,\dots,m^{\prime}$. If we use this notation we have
\begin{align}
U(\beta^{\prime}) &= \sum_{i=1}^{m^{\prime}} \sum_{j=1}^{m} \, \sum_{\nu=1}^{n} p(N_{\nu},\beta^{\prime}_{(i,j)})U(\beta^{\prime},(i,j),\Gamma^{*}_{\beta^{\prime}},\nu) \nonumber \\
&= \sum_{i=1}^{m^{\prime}} \sum_{j=1}^{m} \sum_{\nu=1}^{n} p(\beta_{i}|\alpha_{j}) p(N_{\nu},\alpha_{j})U(\beta^{\prime},j,\Gamma^{*}_{\beta^{\prime}},\nu) \nonumber \\
&= \sum_{j=1}^{m} \sum_{\nu=1}^{n} p(N_{\nu},\alpha_{j})U(\beta^{\prime},j,\Gamma^{*}_{\beta^{\prime}},\nu).\label{eq:bprime}
\end{align}
We now need to relate the last line to $U(\alpha)$.
The POVM $\alpha$ can also be reconstructed from $\beta^{\prime}$ by combining all outcomes $(i,j)$ with the same second index $j$. Since these are POVM elements are proportional, the equality condition in \textbf{C3} along with \textbf{C1} implies
\begin{equation}
U(\beta^{\prime},j,\Gamma^{*}_{\beta^{\prime}},\nu) = U(\alpha,j,\Gamma^{\prime}_{\alpha},\nu),\label{eq:btoa}
\end{equation}
where $\Gamma^{\prime}_{\alpha}$ is the decision strategy induced from $\Gamma^{*}_{\beta^{\prime}}$
\begin{equation}
\Gamma^{\prime}_{\alpha}(j) = \Gamma^{*}_{\beta^{\prime}}(i,j) \quad \text{for any } i=1,\dots,m^{\prime}.
\end{equation}
From \eqref{eq:btoa} and \eqref{eq:bprime}, we have
\begin{align}
U(\beta^{\prime}) &= \sum_{j=1}^{m} \sum_{\nu=1}^{n} p(N_{\nu},\alpha_{j})U(\alpha,j,\Gamma^{\prime}_{\alpha},\nu) \nonumber \\
&\leq \max_{\Gamma\in \mathcal{S}_{\Gamma}} \sum_{j=1}^{m} \sum_{\nu=1}^{n} p(N_{\nu},\alpha_{j})U(\alpha,j,\Gamma^{\prime},\nu) \nonumber \\
&= U(\alpha).
\end{align}
This with \eqref{eq:ub} implies $U(\beta) \leq U(\alpha)$.
\end{proof}

The probability of success \eqref{eq:decsucc} and the total confidence \eqref{eq:confdec} are resource monotones.
This can be checked directly, however here we will show that the utility functions which generate these two functions satisfy all three conditions stated above. 

For minimal error state discrimination, the utility function was simply 
\begin{equation}
U(\epsilon,y,\Gamma,\nu) = \delta_{\Gamma(y),\nu}. 
\end{equation}
Now if $\Gamma(i)=\Gamma(j)$ then we have $U(\epsilon,i,\Gamma,\nu) = U(\epsilon,j,\Gamma,\nu)$,  hence we immediately have that this utility function satisfies conditions \textbf{C1} and \textbf{C3}. 
To see that \textbf{C2} is satisfied, recall that we showed that the best decision strategy was always to choose the state with the highest posterior probability after the measurement. But if two POVM elements are proportional, then all posterior probabilities are identical. This can be seen in \eqref{eq:maxconf}, where multiplying $A$ by a constant doesn't change the posterior distribution. 
Hence for proportional elements $A_{i}$ and $A_{j}$, the best decision strategy will always be to assign $\Gamma(i)=\Gamma(j)$.

Maximal confidence measurements had a slightly more complicated utility function
\begin{equation}
U(\epsilon,y,\Gamma,\nu) = \frac{\delta_{\Gamma(y),\nu}}{p(D_{\Gamma(y)})}. 
\end{equation}
Once again, we can immediately see that this utility function meets conditions \textbf{C1} and \textbf{C3}, since if $\Gamma(i) = \Gamma(j)$, then $U(\epsilon,i,\Gamma,\nu) = U(\epsilon,j,\Gamma,\nu)$. To establish \textbf{C2}, recall that the best decisions strategy involved mapping only single measurement to the possible quantum states, and mapping the rest to the inconclusive outcome.
There we saw that mapping two measurement outcomes to the same quantum state would cause a decrease in the total confidence \eqref{eq:mcbestdec}, \emph{unless} the measurement outcomes yielded the same posterior probability, which is certainly true for measurement outcomes corresponding to proportional POVM elements.

\section{Other Utility Functions}
\label{sec:other}

The framework of Bayesian experimental design allows us to consider different utility functions, and thus value quantities other than minimal error or maximal confidence. For example, if we have a utility function of $U\left(\epsilon,y,\Gamma,\nu\right) = \log\left(p\left(N_{\nu}|\epsilon_{y}\right)\right)$, then the best POVM $\epsilon^\star$ is
\begin{align}
\epsilon^{\star} &= \argmax_{\epsilon\in\mathcal{S}_{\epsilon}} \,\max_{\Gamma\in\mathcal{S}_{\Gamma}} \sum_{y=1}^{m} \sum_{\nu=1}^{n} p\left(N_{\nu},\epsilon_{y}\right)\log\left(p\left(N_{\nu}|\epsilon_{y}\right)\right) \nonumber \\
                 &= \argmax_{\epsilon\in\mathcal{S}_{\epsilon}}  -H\left(N|\epsilon\right)
                   = \argmin_{\epsilon\in\mathcal{S}_{\epsilon}}  H\left(N|\epsilon\right). \label{eq:condinfo}
\end{align}
So the best POVM is the one that minimises the conditional entropy of the variable $N$, given $\epsilon$. The utility function here doesn't depend at all on Bob's decision. Rather, it only values the ability of Bob's measurement to extract the most amount of information on average.

If we add a constant to the utility function we do not change which measurement is optimal. Let us consider 
\begin{align}
U\left(\epsilon,y,\Gamma,\nu\right) &= \log\left(p\left(N_{\nu}|\epsilon_{y}\right)\right) - \sum_{\nu=1}^{n} p\left(N_{\nu}\right)\log\left(p\left(N_{\nu}\right)\right) \nonumber \\
&= \log\left(p\left(N_{\nu}|\epsilon_{y}\right)\right) + H\left(N\right),
\end{align}
where $H\left(N\right)$ is the Shannon entropy of the probability distribution of variable $N$. Now the best POVM is
\begin{align}
\epsilon^{\star} &= \argmax_{\epsilon\in\mathcal{S}_{\epsilon}} \max_{\Gamma\in\mathcal{S}_{\Gamma}} \sum_{y=1}^{m} \sum_{\nu=1}^{n} p\left(N_{\nu},\epsilon_{y}\right)\log\left(p\left(N_{\nu}|\epsilon_{y}\right)\right) + H\left(N\right) \nonumber \\
&= \argmax_{\epsilon\in\mathcal{S}_{\epsilon}} I\left(N:\epsilon\right), \label{eq:mutualinfo}
\end{align}
where $I\left(N:\epsilon\right)$ is the mutual information between variables $N$ and $\epsilon$. This is a common measure of the information gained from a quantum measurement. 
Since $H(N)$ a constant term of Bob's choice of POVM, the same POVM which minimises the conditional entropy, will maximise the mutual information.

We can use the conditions described in the previous section to show that $-H(N|\epsilon)$ and $I(N:\epsilon)$ are resource monotones. 
Since the utility function is a function of the posterior distribution $p(N_{\nu}|\epsilon_{y})$, it is constant over POVM elements which are proportional to each other, meeting condition \textbf{C1}. Condition \textbf{C2} is met trivially, since this utility function is independent of the decision strategy.
Condition \textbf{C3} is less trivial. From Bayes' rule, we have that
\begin{equation}
p(N_{\nu}|\alpha_{i}+\alpha_{j}) = \frac{p(\alpha_{i})p(N_{\nu}|\alpha_{i})+p(\alpha_{j})p(N_{\nu}|\alpha_{j})}{p(\alpha_{i}+\alpha_{j})}.\label{eq:confavg}
\end{equation}
Since the function $x \log(x)$ is convex we have the inequality
\begin{align}
p&(N_{\nu}|\alpha_{i}+\alpha_{j})\log(p(N_{\nu}|\alpha_{i}+\alpha_{j})) \nonumber \\
&\leq \frac{p(\alpha_{i})}{p(\alpha_{i}+\alpha_{j})}p(N_{\nu}|\alpha_{i})\log(p(N_{\nu}|\alpha_{i})) \nonumber \\
& \quad+ \frac{p(\alpha_{j})}{p(\alpha_{i}+\alpha_{j})}p(N_{\nu}|\alpha_{j})\log(p(N_{\nu}|\alpha_{j})).
\end{align}
Again applying Bayes' rule to all three coefficients, we can re-write this inequality as
\begin{align}
\log&(p(N_{\nu}|\alpha_{i}+\alpha_{j})) \nonumber \\
& \leq \frac{p(\alpha_{i}|N_{\nu})\log(p(N_{\nu}|\alpha_{i}))+p(\alpha_{j}|N_{\nu})\log(p(N_{\nu}|\alpha_{j}))}{p(\alpha_{i}+\alpha_{j}|N_{\nu})},
\end{align}
which is exactly condition \textbf{C3}.
From \eqref{eq:confavg} we see that if $A_{i}$ and $A_{j}$ are proportional, then
\begin{equation}
p(N_{\nu}|\alpha_{i}+\alpha_{j}) = p(N_{\nu}|\alpha_{i}) = p(N_{\nu}|\alpha_{j}),
\end{equation}
which is the equality condition.\\

The quantities $-H(N|\epsilon)$ and $I(N:\epsilon)$ are resource monotones for the stochastic mixing of POVMs. 
This means that if $\alpha\geq \beta$ under stochastic mixing, then for any resource monotone $f$, $f(\alpha)\geq f(\beta)$.
This implies that there always exists an optimal POVM for these functions that is constructed from only rank-1 POVM elements, since any POVM can be generated from a stochastic mixture of a rank-1 POVM.
To see this, consider the POVM $\left(A_{1},\dots,A_{m}\right)$ and consider the spectral decomposition
\begin{equation}
A_{i} = \sum_{j=1}^{d} \lambda_{j}^{i}\ketbra{\lambda_{j}^{i}}{\lambda_{j}^{i}},
\end{equation}
for each $i=1,\dots,m$. The collection of all rank-1 elements
\begin{equation}
\left(\lambda_{1}\ketbra{\lambda_{1}^{1}}{\lambda_{1}^{1}},\dots,\lambda_{d}^{m}\ketbra{\lambda_{d}^{m}}{\lambda_{d}^{m}}\right).
\end{equation}
is itself a POVM. The original POVM $\left(A_{1},\dots,A_{m}\right)$ can easily be seen as a stochastic mixture of this POVM.

This argument also holds for the probability of success \eqref{eq:decsucc} and the total confidence \eqref{eq:confdec}. Since their utility functions also satisfy \textbf{C1}-\textbf{C3}, this also implies that the optimal POVM in those cases would consist of rank-1 POVM elements. 
However, as mentioned earlier, one need not restrict to POVMs with only rank-1 elements in the context of minimal error and maximal confidence state discrimination.
This is because, for a given decision strategy $\Gamma$, we can create a new POVM by combining the elements of a rank-1 POVM corresponding to each decision $\nu$. This new POVM will have the same probability of deciding on $\nu$.
This does not hold for the mutual or conditional entropy since the entropy is not a function of the probability of the decision, but the probability of the measurement outcome.

Indeed, one might want to consider the conditional entropy and mutual information based on Bob's decision rather than his measurement outcome. That is 
\begin{equation}
U\left(\epsilon,y,\Gamma,\nu\right) = \delta_{\Gamma(y),\nu}\log\left(p\left(N_{\nu}|D_{\Gamma(y)}\right)\right).
\end{equation}
This leads to the following optimisation,
\begin{align}
&\argmax_{\epsilon\in\mathcal{S}_{\epsilon}}\; \max_{\Gamma\in\mathcal{S}_{\Gamma}}\; \sum_{\nu=1}^{n} \sum_{y=1}^{m} p\left(N_{\nu},\epsilon_{y}\right) \delta_{\Gamma(y),\nu}\log\left(p\left(N_{\nu}|D_{\Gamma(y)}\right)\right). \nonumber \\
&= \argmax_{\epsilon\in\mathcal{S}_{\epsilon}}\; \max_{\Gamma\in\mathcal{S}_{\Gamma}}\; \sum_{y=1}^{m} p\left(N_{\Gamma(y)},\epsilon_{y}\right) \log\left(p\left(N_{\Gamma(y)}|D_{\Gamma(y)}\right)\right) \nonumber \\
&= \argmax_{\epsilon\in\mathcal{S}_{\epsilon}}\; \max_{\Gamma\in\mathcal{S}_{\Gamma}}\; \sum_{d=1}^{n} p\left(N_{d},D_{d}\right) \log\left(p\left(N_{d}|D_{d}\right)\right) \nonumber \\
&= \argmax_{\epsilon\in\mathcal{S}_{\epsilon}}\; \max_{\Gamma\in\mathcal{S}_{\Gamma}}\; -H\left(N|D\right)_{\epsilon}.
\end{align}
We can likewise add a constant $H\left(N\right)$ to convert this to the mutual information $I\left(N:D\right)_{\epsilon}$. Thus in this case we want to optimise the correlations between the decision Bob makes and the state Alice sends.

\section{Multi-Objective Utility Functions}
\label{sec:multiobj}

In this section we look at multi-objective utility functions, which are useful when one would like to consider multiple optimisations in sequence. 
This would arise when a single utility function does not choose a single optimal measurement. 
This can arise in the context of both maximal confidence and minimal error state discrimination.
Here we allow the utility function to be vector valued, equipped with the dictionary order (described in the early paragraphs of section~\ref{sec:BED}).

\subsection{Maximal Confidence State Discrimination}

Recall that in the construction of the maximal confidence measurement \eqref{eq:maxconf}, rescaling the POVM elements by a constant does not change the confidence in that measurement outcome, so there is a continuous family of measurements that all have the maximal possible confidence. However rescaling POVM elements can change the probability of the inconclusive outcome. 
Since the inconclusive outcome is undesirable, one therefore may wish to choose the maximal confidence measurement which minimises the probability of the inconclusive outcome. 
Let us show a simple example of this.

\begin{example}\label{ex:unambig}
Let Alice use two possible qutrit states $\rho_{1}$ and $\rho_{2}$ to encode her messages, and she will send the messages with equal probability $q_{1}=q_{2}=\frac{1}{2}$. Using the notation that $\ket{\pm} = \frac{1}{2}\left(\ket{1}\pm\ket{2}\right)$, let
\begin{equation}
\begin{matrix}
\rho_{1} = \frac{1}{2}\left(\ketbra{0}{0}+\ketbra{1}{1}\right), \\ 
\\
\rho_{2} = \ketbra{+}{+}.
\end{matrix}
\end{equation}
These two states can be distinguished unambiguously, for example, if we choose $A_{1} = a_{1}\ketbra{-}{-}$ and $A_{2} = a_{2}\ketbra{2}{2}$ then $p\left(\rho_{1}|A_{1}\right) = p\left(\rho_{2}|A_{2}\right) = 1$. 
There is also the inconclusive outcome represented by $A_{0} = \mathbbm{1} - A_{1} - A_{2}$, which is only positive semi-definite when $0\leq a_{1}\leq 1$ and $0\leq a_{2}\leq 2-2/(2-a_{1})$.
Any choice of $a_{1},a_{2}$ within these parameters is a valid POVM which will unambiguously distinguish between $\rho_{1}$ and $\rho_{2}$, but clearly some choices are better than others; if both $a_{1}$ and $a_{2}$ are vanishingly small then the inconclusive outcome will almost always register. 
Thus parameters are ideally chosen to minimise the probability of the inconclusive outcome $p\left(A_{0}\right)$.
In this case that is achieved by making $a_{1}$ as small as possible and $a_{2}$ as large as possible, for which $p\left(A_{0}\right)$ approaches $\frac{3}{4}$.

In fact, one can do even better since $A_{1}$ and $A_{2}$ did not fully characterise the set of maximal confidence measurements. One can choose the projective measurement $B_{1} = \ketbra{0}{0}$, $B_{2} = \ketbra{2}{2}$; which implies that $B_{0} = \ketbra{1}{1}$. Here the probability of the inconclusive outcome $p\left(B_{0}\right)=\frac{1}{2}$, and cannot be improved by rescaling $B_{1}$ and $B_{2}$. This demonstrates that the family of maximal confidence measurements won't always be related to each other by rescaling.
\end{example}

This secondary optimisation can be easily incorporated into Bayesian experimental design by considering utility functions whose image lies in $\mathbb{R}^{2}$, ordered under the dictionary order.
For example, Bob may want the maximal confidence measurement which, within this family, maximises the probability of success \eqref{eq:decsucc}. This happens to be equivalent to choosing the maximal confidence with the smallest probability of yielding the inconclusive outcome.
For any maximally confident POVM, the probability of success is
\begin{equation}
\sum_{\nu=1}^{n} p\left(N_{\nu}|D_{\nu}\right)p\left(D_{\nu}\right) = \sum_{\nu=1}^{n} c_{\nu} \, p\left(D_{\nu}\right),
\end{equation}
where $c_{\nu}=p\left(\rho_{\nu}|D_{\nu}\right)$ are constants within the set of maximally confident POVMs. Define $c_{0}$ as the positive number for which
\begin{equation}
 \sum_{\nu=1}^{n} c_{\nu} \, p\left(D_{\nu}\right) + c_{0}\, p\left(D_{0}\right) = 1.
\end{equation}
Thus we see that maximising the probability of success is equivalent to minimising the probability of the inconclusive outcome
\begin{equation}
 \sum_{j=1}^{n} c_{j} \, p\left(D_{j}\right) = 1- c_{0}\, p\left(D_{0}\right).
\end{equation}

If Bob desires a maximal confidence measurement which has the smallest probability of acquiring the inconclusive outcome, then Bob's utility function is
\begin{equation}
U(\epsilon,y,\Gamma,\nu) = \left(\frac{\delta_{\Gamma\left(y\right),\nu}}{p\left(D_{\Gamma\left(y\right)}\right)},1-\delta_{\Gamma\left(y\right),0}\right),
\end{equation}
which encodes that Bob first wants to maximise the confidence, and only then search for the POVM with the lowest inconclusive probability.

The optimal POVM now satisfies
\begin{align}
  \epsilon^{\star} &= \argmax_{\epsilon\in \mathcal{S}_\epsilon}\, \max_{\Gamma\in\mathcal{S}_{\Gamma}}\; (U_{1}(\epsilon,\Gamma),U_{2}(\epsilon,\Gamma)) 
\end{align}
where
\begin{subequations}
\begin{align}
U_{1}(\epsilon,\Gamma) &= \sum_{y=1}^{m} \sum_{\nu=1}^{n}
  \frac{p\left(N_{\nu},\epsilon_{y}\right)\delta_{\Gamma\left(y\right),\nu}}{p\left(D_{\Gamma\left(y\right)}\right)} \nonumber \\
  &= \sum_{\nu=1}^{n} p\left(N_{\nu}|D_{\nu}\right).\\
U_{2}(\epsilon,\Gamma) &=  \sum_{y=1}^{m} \sum_{\nu=1}^{n}  p\left(N_{\nu}, \epsilon_{y}\right)\left(1\!-\!\delta_{\Gamma\left(y\right),0}\right)\nonumber \\
&= 1\!-\!p\left(D_{0}\right).
\end{align}
\end{subequations}
Maximising under the dictionary order maximises each expression in order,
subject to all previous expressions being maximised. Therefore, in this case the
best measurement is the maximal confidence measurement that minimises the
probability of the null decision.

\subsection{Minimal Error State Discrimination}

There are examples for which it would be useful to perform a secondary optimisation on minimal error measurements.

In \cite{hunter_measurement_2003}, the author derives the following result: if there exists a $\rho_{m}$ in
Alice's alphabet which satisfies
\begin{equation}
    q_{m}\rho_{m} - q_{k}\rho_{k}\geq 0,\quad \text{for all } k, \label{eq:alwaysm}
\end{equation}
then the optimal measurement is represented by the POVM, whose $j^{\text{th}}$
element is $\delta_{j,m}\mathbbm{1}$. In other words, the optimal measurement
is to perform no measurement and simply guess the state $\rho_{m}$.
The author shows that this set of inequalities can be satisfied with non-trivial alphabets.

This appears counter-intuitive; measurements which provide some information should always be better (as we need not act on it). In fact, if \eqref{eq:alwaysm} is satisfied, all measurements have the same (and hence optimal) probability of success after optimising over the decision strategy. From \eqref{eq:alwaysm} we have, for all $k$, and for any positive semi-definite operator $A_{\nu}\leq \mathbbm{1}$,
\begin{align}
\tr\left(\left(q_{m}\rho_{m} - q_{k}\rho_{k}\right)A_{\nu}\right) &\geq 0, \nonumber \\
\tr\left(q_{m}\rho_{m}A_{\nu}\right) &\geq \tr\left(q_{k}\rho_{k}A_{\nu}\right), \nonumber \\
P\left(\rho_{m}|A_{\nu}\right)P\left(A_{\nu}\right) &\geq P\left(\rho_{k}|A_{\nu}\right)P\left(A_{\nu}\right), \nonumber \\
P\left(\rho_{m}|A_{\nu}\right) &\geq P\left(\rho_{k}|A_{\nu}\right).
\end{align}
That is, the state $\rho_{m}$ will always be the state with the largest posterior probability, no matter the measurement or the measurement result.
Therefore, regardless of what measurement Bob does, his best decision is to always choose that Alice sent state $\rho_{m}$. 
For these kinds of alphabets and prior probabilities any measurement device is as good as any other in the context of minimal error state discrimination. This includes the trivial measurement, represented with a single non-zero POVM element, and physically corresponds to Bob doing no measurement and simply deciding that Alice sent $\rho_{m}$ every time.

But each non-trivial measurement still provides information about Alice's state though and further measurements may lead to Bob changing his decision. 
To see this, we can generalise the condition~\eqref{eq:alwaysm} to the case of repeated messages.

\begin{proposition}
  Suppose Alice sends $d$ copies of the same state. After Bob performs a local
  measurement on each individual state sent, then $\rho_{m}$ will always have the largest posterior probability if and only if
    \begin{equation}
    \sqrt[d]{q_{m}}\rho_{m} - \sqrt[d]{q_{k}}\rho_{k}\geq 0 \quad \text{for all } k. \label{eq:dcopies}
  \end{equation}
\end{proposition}
\begin{proof}
    Suppose \eqref{eq:dcopies} holds; then for any positive semi-definite operator $A_{j}$,
    \begin{equation}
        \sqrt[d]{q_{m}}\Tr\left(\rho_{m}A_{j}\right)\geq \sqrt[d]{q_{k}}\Tr\left(\rho_{k}A_{j}\right),
    \end{equation}
    for all $k$. This inequality holds if we take the product of $d$ such terms. Let $A_{j}^{\left(i\right)}$ denote that result $A_{j}$ was measured on the $i^{\text{th}}$ system.
    \begin{align}
        \prod_{j=1}^{d}\sqrt[d]{q_{m}}\Tr\left(\rho_{m}A_{j}^{(j)}\right)&\geq \prod_{j=1}^{d}\sqrt[d]{q_{k}}\Tr\left(\rho_{k}A_{j}^{(j)}\right), \nonumber \\
        \Rightarrow 	q_{m}\prod_{j=1}^{d}\Tr\left(\rho_{m}A_{j}^{(j)}\right)&\geq q_{k}\prod_{j=1}^{d}\Tr\left(\rho_{k}A_{j}^{(j)}\right),
    \end{align}
    so $P\left(\rho_{m}|A_{1}^{\left(1\right)}\dots A_{d}^{\left(d\right)}\right) \geq P\left(\rho_{k}|A_{1}^{\left(1\right)}\dots A_{d}^{\left(d\right)}\right)$,
    for all $k$. If \eqref{eq:dcopies} doesn't hold, then there exists a positive semi-definite operator $B$ and state $\rho_{k}$ such that
    \begin{equation}
        \sqrt[d]{q_{m}}\Tr\left(\rho_{m}B\right) < \sqrt[d]{q_{k}}\Tr\left(\rho_{k}B\right).
    \end{equation}
    Taking the $d^{\text{th}}$ power of both sides,
    \begin{equation}
        q_{m}\Tr\left(\rho_{m}B\right)^{d} < q_{k}\Tr\left(\rho_{k}B\right)^{d},
    \end{equation}
    and so $P\left(\rho_{m}|B^{\left(1\right)}\dots
    B^{\left(d\right)}\right)<P\left(\rho_{k}|B^{\left(1\right)}\dots
    B^{\left(d\right)}\right)$; there exists a measurement record which would cause
    Bob to decide that Alice sent $\rho_{k}$ where $k\neq m$.
\end{proof}

Clearly, \eqref{eq:alwaysm} is the case where $d=1$.

\begin{figure}[t]
  \includegraphics[width=0.4\textwidth]{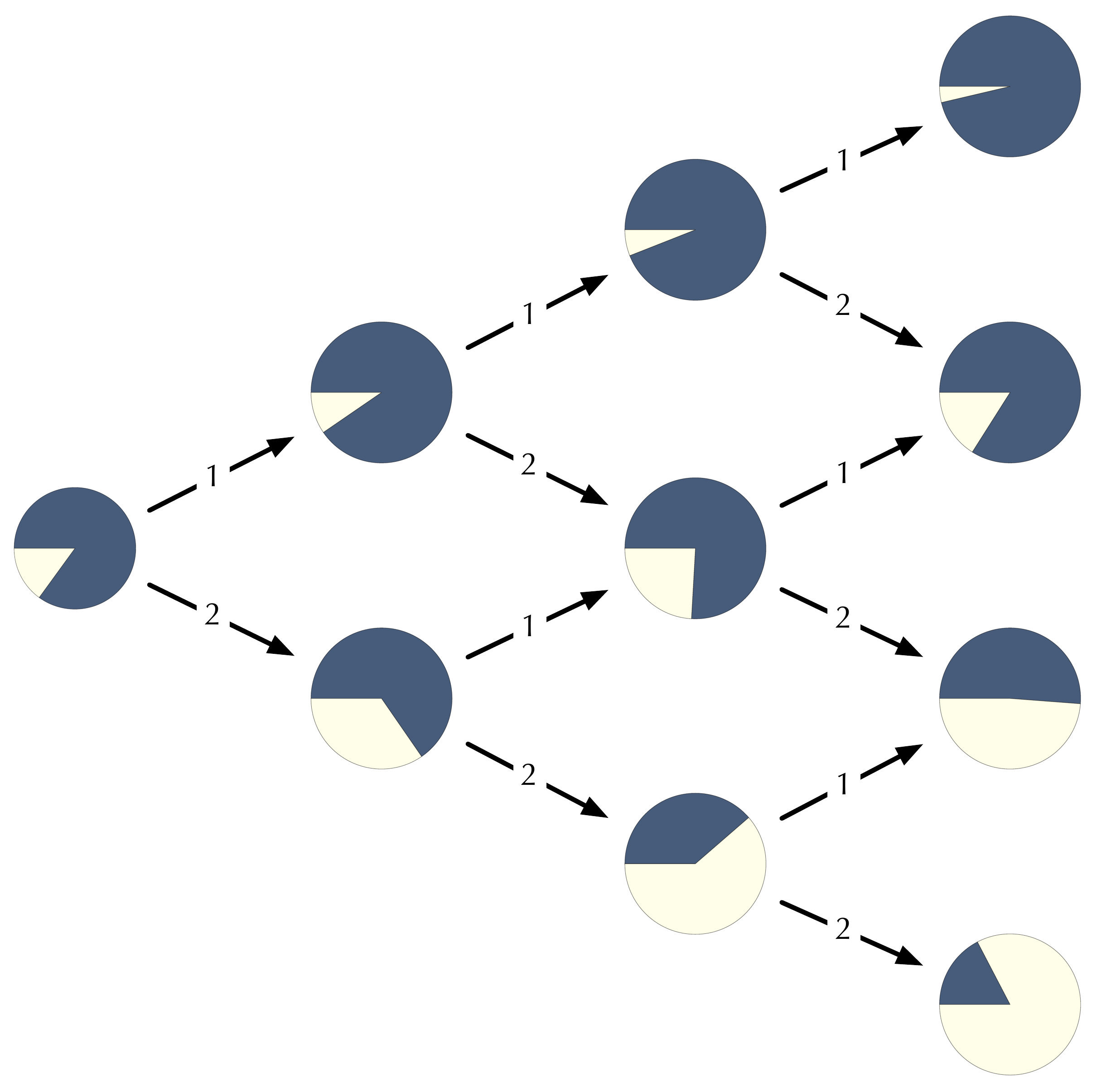}
  \centering
  \caption{Even when an alphabet and prior probabilities satisfy the conditions in \eqref{eq:alwaysm}, repeated measurements can provide enough information to make Bob not choose $\rho_{m}$. In this example Alice's alphabet consists of states $\rho_{1} = \frac{5}{6}\ketbra{0}{0}+\frac{1}{6}\ketbra{1}{1}$ and $\rho_{2} = \ketbra{+}{+}$. The prior probabilities are $q_{1}=0.85$ and $q_{2}=0.15$. The posterior probabilities for $q_{1}$ (dark) and $q_{2}$ (light) over the course of three repeated measurements using POVM $\set{A_{1} = \ketbra{0}{0},A_{2}=\ketbra{1}{1}}$. We have $q_{1}\rho_{1} - q_{2}\rho_{2} \geq 0$, so no single measurement could cause Bob to choose $\rho_{2}$ over $\rho_{1}$. However since $\sqrt{q_{1}}\rho_{1} - \sqrt{\rho_{2}}\rho_{2}\not\geq 0$, there exists a two-measurement record (in this case, receiving the outcome $A_{2}$ on the first two measurements) which can cause Bob to choose $\rho_{2}$.}
  \label{fig:repeatedmeasurements}
\end{figure}

Non-trivial measurements still provide Bob with information, modifying his
posterior knowledge of what Alice sent; however no measurement exists which
could overturn his bias that $\rho_{m}$ is the most probable state, as reflected
in the prior probabilities. 
Indeed it can be argued that performing no measurement is the least useful, since if Alice were to send multiple copies, non-trivial measurements exist which could indeed cause Bob to change his mind, as illustrated in fig~\ref{fig:repeatedmeasurements}.

Since all states contain the same probability of error, it is sensible to perform a second optimisation to select the optimal measurement. 
The condition \eqref{eq:alwaysm} will enforce that $p(N_{m} | \epsilon_{y}) \geq p(N_{k} | \epsilon_{y})$ for all $k=1,\dots,n$, and for any outcome $y$ of any measurement $\epsilon$. 
Therefore, since the optimal decision strategy of minimal error state discrimination is to choose the state with the highest posterior probability, then we know that $p(D_{m}) = 1$. 
This means that using utility functions in the secondary optimisation which depend on Bob's decision will be equally unhelpful. 
A more useful candidate is the mutual information described in section~\ref{sec:other}; which measures how much information the measurement gains about which state Alice sent. Let our multi-objective utility function be
\begin{equation}
U(\epsilon,y,\Gamma,\nu) =  \left(\delta_{\Gamma(y),\nu}\, ,\,\log\left(p\left(N_{\nu}|\epsilon_{y}\right)\right) + H(N)\right).
\end{equation}
The measurement $\epsilon^{\star}$ which maximises this utility function will satisfy
\begin{equation}
\epsilon^{\star} = \argmax_{\epsilon\in \mathcal{S}_\epsilon}\, \max_{\Gamma\in\mathcal{S}_{\Gamma}}\; \left(\sum_{\nu=1}^{n} p(D_{\nu}|N_{\nu})p(N_{\nu})\, ,\, I(N:\epsilon)\right).
\end{equation}
Of course the condition \eqref{eq:alwaysm} implies that $\sum_{\nu=1}^{n} p(D_{\nu}|N_{\nu})p(N_{\nu}) = p(N_{m})$, assuming the optimal decision strategy of choosing the state with the largest posterior distribution. Therefore this utility function will select the measurement $\epsilon^{\star}$ that maximises the mutual information $I(N:\epsilon)$, that is the state which extracts the most information, along with the decision strategy  of choosing the state with the largest posterior distribution.
This is a more sensible choice for the optimal measurement than simply the trivial measurement.

\section{Conclusion} \label{sec:conclusion}

In this paper we showed how quantum state discrimination fits neatly into Bayesian experimental design.
Bayesian experimental design is a very general framework designed to give a methodology to finding an optimal experiment. 
Since quantum state discrimination is typically concerned with finding the optimal measurement to discriminate between quantum states, it is unsurprising that quantum state discrimination can be cast into this framework.
From this perspective the standard paradigms of minimal error state discrimination and maximal confidence state discrimination are the optimal measurements corresponding to different utility functions.

The utility function encodes the priorities of the experimenter. In minimal error state discrimination, the priority is to make as few mistakes as possible. In maximal confidence state discrimination, the priority is having the maximal possible assurance that each measurement outcome is accurate.
This suggests that one can consider many other priorities by changing the utility function.
As an example, we showed how any the mutual information between the distribution of the states being discriminated and the measurement outcomes arises from a utility function. 

Quantum state discrimination is beginning to be seen as useful in the context of resource theories.
In particular, discrimination tasks can be used to measure the quality of a quantum resource \cite{Takagi2019}.
This is usually performed in the context of minimal error state discrimination, where it is known that probability of success \eqref{eq:avesuccess} can fully describe the resource theory of quantum measurements \cite{Skrzypczyk2019,Guff2019}.
However, the set of all state discimination tasks is an uncountable set of monotones, making calculation difficult. In this paper we showed that the total confidence is also a resource monotone.
In fact we provides simple conditions on the utility function which guarantee that the averaged utility function $U(\epsilon)$ \eqref{eq:Uep} is a resource monotone.
This could lead to a wider class of functions that can be used to classify the resource theory of quantum measurement; potentially leading to a finite family of monotones that fully describes this resource theory.

If the average utility $U(\epsilon)$ can be shown to be a resource monotone, then we showed in section~\eqref{sec:other} that this implies that $U(\epsilon)$ there exists an optimal POVM consisting only of rank-1 elements.
This result might be useful in restricting the set of possible measurements one needs to search over in trying to find the optimum.

In order to incorporate quantum state discrimination into Bayesian experimental design, we generalised the framework in two directions.
Firstly, we performed a maximisation over all \emph{decision strategies}, instead simply over individual decisions.
This was necessary to describe maximal confidence state discrimination, whose utility function depends upon the entire decision strategy. 
When the utility function depends only on the explicit decision made, the framework reduces to the standard expression.
Secondly, we allowed the utility function to be $\mathbb{R}^{n}$-valued, equipped with the dictionary order.
This allows one to optimise multiple priorities in turn, subject to all previous priorities being optimised.
We showed the benefit of this multi-objective optimisation in two cases. 
The first is finding the best out of the usual family of maximal confidence measurements. 
The second was the case in which all measurements have the same probability of success. 
In this case minimal error state discrimination is entirely unhelpful.
A secondary optimisation of the mutual information will select the experiment which on average extracts the most information.

Now that quantum state discrimination has been cast into the framework of Bayesian experimental design, the analytic or numerical methods used to solve Bayesian experimental design optimisation problems may be useful in finding the optimal measurement for many quantum state discrimination tasks.

\begin{acknowledgements}
We would like to thank a referee in the first round of review of this manuscript which suggested the argument for combining POVM elements corresponding to the same decision without changing the probability of success.
  This research was funded in part by the Australian Research Council Centre of
  Excellence for Engineered Quantum Systems (Project number CE170100009). 
\end{acknowledgements}

\bibliography{references}

\begin{thebibliography}{23}%
\makeatletter
\providecommand \@ifxundefined [1]{%
 \@ifx{#1\undefined}
}%
\providecommand \@ifnum [1]{%
 \ifnum #1\expandafter \@firstoftwo
 \else \expandafter \@secondoftwo
 \fi
}%
\providecommand \@ifx [1]{%
 \ifx #1\expandafter \@firstoftwo
 \else \expandafter \@secondoftwo
 \fi
}%
\providecommand \natexlab [1]{#1}%
\providecommand \enquote  [1]{``#1''}%
\providecommand \bibnamefont  [1]{#1}%
\providecommand \bibfnamefont [1]{#1}%
\providecommand \citenamefont [1]{#1}%
\providecommand \href@noop [0]{\@secondoftwo}%
\providecommand \href [0]{\begingroup \@sanitize@url \@href}%
\providecommand \@href[1]{\@@startlink{#1}\@@href}%
\providecommand \@@href[1]{\endgroup#1\@@endlink}%
\providecommand \@sanitize@url [0]{\catcode `\\12\catcode `\$12\catcode
  `\&12\catcode `\#12\catcode `\^12\catcode `\_12\catcode `\%12\relax}%
\providecommand \@@startlink[1]{}%
\providecommand \@@endlink[0]{}%
\providecommand \url  [0]{\begingroup\@sanitize@url \@url }%
\providecommand \@url [1]{\endgroup\@href {#1}{\urlprefix }}%
\providecommand \urlprefix  [0]{URL }%
\providecommand \Eprint [0]{\href }%
\providecommand \doibase [0]{http://dx.doi.org/}%
\providecommand \selectlanguage [0]{\@gobble}%
\providecommand \bibinfo  [0]{\@secondoftwo}%
\providecommand \bibfield  [0]{\@secondoftwo}%
\providecommand \translation [1]{[#1]}%
\providecommand \BibitemOpen [0]{}%
\providecommand \bibitemStop [0]{}%
\providecommand \bibitemNoStop [0]{.\EOS\space}%
\providecommand \EOS [0]{\spacefactor3000\relax}%
\providecommand \BibitemShut  [1]{\csname bibitem#1\endcsname}%
\let\auto@bib@innerbib\@empty
\bibitem [{\citenamefont {Chefles}(2000)}]{chefles_review_2000}%
  \BibitemOpen
  \bibfield  {author} {\bibinfo {author} {\bibfnamefont {A.}~\bibnamefont
  {Chefles}},\ }\href {\doibase 10.1080/00107510010002599} {\bibfield
  {journal} {\bibinfo  {journal} {Cont. Phys.}\ }\textbf {\bibinfo {volume}
  {41}},\ \bibinfo {pages} {401} (\bibinfo {year} {2000})}\BibitemShut
  {NoStop}%
\bibitem [{\citenamefont {Barnett}\ and\ \citenamefont
  {Croke}(2009{\natexlab{a}})}]{barnett_review_2009}%
  \BibitemOpen
  \bibfield  {author} {\bibinfo {author} {\bibfnamefont {S.~M.}\ \bibnamefont
  {Barnett}}\ and\ \bibinfo {author} {\bibfnamefont {S.}~\bibnamefont
  {Croke}},\ }\href {\doibase 10.1364/AOP.1.000238} {\bibfield  {journal}
  {\bibinfo  {journal} {Adv. Opt. Phot.}\ }\textbf {\bibinfo {volume} {1}},\
  \bibinfo {pages} {238} (\bibinfo {year} {2009}{\natexlab{a}})}\BibitemShut
  {NoStop}%
\bibitem [{\citenamefont {Bae}\ and\ \citenamefont
  {Kwek}(2015)}]{bae_review_2015}%
  \BibitemOpen
  \bibfield  {author} {\bibinfo {author} {\bibfnamefont {J.}~\bibnamefont
  {Bae}}\ and\ \bibinfo {author} {\bibfnamefont {L.-C.}\ \bibnamefont {Kwek}},\
  }\href {\doibase 10.1088/1751-8113/48/8/083001} {\bibfield  {journal}
  {\bibinfo  {journal} {J. Phys. A.}\ }\textbf {\bibinfo {volume} {48}},\
  \bibinfo {pages} {083001} (\bibinfo {year} {2015})}\BibitemShut {NoStop}%
\bibitem [{\citenamefont {Takagi}\ and\ \citenamefont
  {Regula}(2019)}]{Takagi2019}%
  \BibitemOpen
  \bibfield  {author} {\bibinfo {author} {\bibfnamefont {R.}~\bibnamefont
  {Takagi}}\ and\ \bibinfo {author} {\bibfnamefont {B.}~\bibnamefont
  {Regula}},\ }\href {\doibase 10.1103/PhysRevX.9.031053} {\bibfield  {journal}
  {\bibinfo  {journal} {Phys. Rev. X.}\ }\textbf {\bibinfo {volume} {9}},\
  \bibinfo {pages} {031053} (\bibinfo {year} {2019})},\ \Eprint
  {http://arxiv.org/abs/1901.08127} {arXiv:1901.08127} \BibitemShut {NoStop}%
\bibitem [{\citenamefont {Skrzypczyk}\ and\ \citenamefont
  {Linden}(2019)}]{Skrzypczyk2019}%
  \BibitemOpen
  \bibfield  {author} {\bibinfo {author} {\bibfnamefont {P.}~\bibnamefont
  {Skrzypczyk}}\ and\ \bibinfo {author} {\bibfnamefont {N.}~\bibnamefont
  {Linden}},\ }\href {\doibase 10.1103/PhysRevLett.122.140403} {\bibfield
  {journal} {\bibinfo  {journal} {Phys. Rev. Lett.}\ }\textbf {\bibinfo
  {volume} {122}},\ \bibinfo {pages} {140403} (\bibinfo {year}
  {2019})}\BibitemShut {NoStop}%
\bibitem [{\citenamefont {Skrzypczyk}\ \emph {et~al.}(2019)\citenamefont
  {Skrzypczyk}, \citenamefont {{\v{S}}upi{\'{c}}},\ and\ \citenamefont
  {Cavalcanti}}]{Skrzypczyk2019a}%
  \BibitemOpen
  \bibfield  {author} {\bibinfo {author} {\bibfnamefont {P.}~\bibnamefont
  {Skrzypczyk}}, \bibinfo {author} {\bibfnamefont {I.}~\bibnamefont
  {{\v{S}}upi{\'{c}}}}, \ and\ \bibinfo {author} {\bibfnamefont
  {D.}~\bibnamefont {Cavalcanti}},\ }\href {\doibase
  10.1103/PhysRevLett.122.130403} {\bibfield  {journal} {\bibinfo  {journal}
  {Phys. Rev. Lett.}\ }\textbf {\bibinfo {volume} {122}},\ \bibinfo {pages}
  {130403} (\bibinfo {year} {2019})}\BibitemShut {NoStop}%
\bibitem [{\citenamefont {Guff}\ \emph {et~al.}(2019)\citenamefont {Guff},
  \citenamefont {McMahon}, \citenamefont {Sanders},\ and\ \citenamefont
  {Gilchrist}}]{Guff2019}%
  \BibitemOpen
  \bibfield  {author} {\bibinfo {author} {\bibfnamefont {T.}~\bibnamefont
  {Guff}}, \bibinfo {author} {\bibfnamefont {N.~A.}\ \bibnamefont {McMahon}},
  \bibinfo {author} {\bibfnamefont {Y.~R.}\ \bibnamefont {Sanders}}, \ and\
  \bibinfo {author} {\bibfnamefont {A.}~\bibnamefont {Gilchrist}},\ }\href
  {http://arxiv.org/abs/1902.08490} {\  (\bibinfo {year} {2019})},\ \Eprint
  {http://arxiv.org/abs/1902.08490} {arXiv:1902.08490} \BibitemShut {NoStop}%
\bibitem [{\citenamefont {Helstrom}(1976)}]{helstrom_book_1976}%
  \BibitemOpen
  \bibfield  {author} {\bibinfo {author} {\bibfnamefont {C.~W.}\ \bibnamefont
  {Helstrom}},\ }\href
  {https://www.elsevier.com/books/quantum-detection-and-estimation-theory/helstrom/978-0-12-340050-5}
  {\emph {\bibinfo {title} {{Quantum detection and estimation theory}}}}\
  (\bibinfo  {publisher} {Academic Press},\ \bibinfo {year} {1976})\BibitemShut
  {NoStop}%
\bibitem [{\citenamefont {Chou}(2004)}]{Chou_mirror_2004}%
  \BibitemOpen
  \bibfield  {author} {\bibinfo {author} {\bibfnamefont {C.-L.}\ \bibnamefont
  {Chou}},\ }\href {\doibase 10.1103/PhysRevA.70.062316} {\bibfield  {journal}
  {\bibinfo  {journal} {Phys. Rev. A}\ }\textbf {\bibinfo {volume} {70}},\
  \bibinfo {pages} {062316} (\bibinfo {year} {2004})}\BibitemShut {NoStop}%
\bibitem [{\citenamefont {Andersson}\ \emph {et~al.}(2002)\citenamefont
  {Andersson}, \citenamefont {Barnett}, \citenamefont {Gilson},\ and\
  \citenamefont {Hunter}}]{Andersson_mirror_2002}%
  \BibitemOpen
  \bibfield  {author} {\bibinfo {author} {\bibfnamefont {E.}~\bibnamefont
  {Andersson}}, \bibinfo {author} {\bibfnamefont {S.~M.}\ \bibnamefont
  {Barnett}}, \bibinfo {author} {\bibfnamefont {C.~R.}\ \bibnamefont {Gilson}},
  \ and\ \bibinfo {author} {\bibfnamefont {K.}~\bibnamefont {Hunter}},\ }\href
  {\doibase 10.1103/PhysRevA.65.052308} {\bibfield  {journal} {\bibinfo
  {journal} {Phys. Rev. A}\ }\textbf {\bibinfo {volume} {65}},\ \bibinfo
  {pages} {052308} (\bibinfo {year} {2002})}\BibitemShut {NoStop}%
\bibitem [{\citenamefont {Barnett}\ and\ \citenamefont
  {Croke}(2009{\natexlab{b}})}]{barnett_conditions_2009}%
  \BibitemOpen
  \bibfield  {author} {\bibinfo {author} {\bibfnamefont {S.~M.}\ \bibnamefont
  {Barnett}}\ and\ \bibinfo {author} {\bibfnamefont {S.}~\bibnamefont
  {Croke}},\ }\href {\doibase 10.1088/1751-8113/42/6/062001} {\bibfield
  {journal} {\bibinfo  {journal} {J. Phys. A}\ }\textbf {\bibinfo {volume}
  {42}},\ \bibinfo {pages} {062001} (\bibinfo {year}
  {2009}{\natexlab{b}})}\BibitemShut {NoStop}%
\bibitem [{\citenamefont {Peres}(1988)}]{Peres_unambiguous_1988}%
  \BibitemOpen
  \bibfield  {author} {\bibinfo {author} {\bibfnamefont {A.}~\bibnamefont
  {Peres}},\ }\href {\doibase 10.1016/0375-9601(88)91034-1} {\bibfield
  {journal} {\bibinfo  {journal} {Phys. Lett. A}\ }\textbf {\bibinfo {volume}
  {128}},\ \bibinfo {pages} {19} (\bibinfo {year} {1988})}\BibitemShut
  {NoStop}%
\bibitem [{\citenamefont {Ivanovic}(1987)}]{Ivanovic_unanbiguous_1987}%
  \BibitemOpen
  \bibfield  {author} {\bibinfo {author} {\bibfnamefont {I.}~\bibnamefont
  {Ivanovic}},\ }\href {\doibase 10.1016/0375-9601(87)90222-2} {\bibfield
  {journal} {\bibinfo  {journal} {Phys. Lett. A}\ }\textbf {\bibinfo {volume}
  {123}},\ \bibinfo {pages} {257} (\bibinfo {year} {1987})}\BibitemShut
  {NoStop}%
\bibitem [{\citenamefont {Dieks}(1988)}]{dieks_overlap_1988}%
  \BibitemOpen
  \bibfield  {author} {\bibinfo {author} {\bibfnamefont {D.}~\bibnamefont
  {Dieks}},\ }\href {\doibase 10.1016/0375-9601(88)90840-7} {\bibfield
  {journal} {\bibinfo  {journal} {Phys. Lett. A}\ }\textbf {\bibinfo {volume}
  {126}},\ \bibinfo {pages} {303} (\bibinfo {year} {1988})}\BibitemShut
  {NoStop}%
\bibitem [{\citenamefont {Jaeger}\ and\ \citenamefont
  {Shimony}(1995)}]{Jaeger_unambiguous_1995}%
  \BibitemOpen
  \bibfield  {author} {\bibinfo {author} {\bibfnamefont {G.}~\bibnamefont
  {Jaeger}}\ and\ \bibinfo {author} {\bibfnamefont {A.}~\bibnamefont
  {Shimony}},\ }\href {\doibase 10.1016/0375-9601(94)00919-G} {\bibfield
  {journal} {\bibinfo  {journal} {Phys. Lett. A}\ }\textbf {\bibinfo {volume}
  {197}},\ \bibinfo {pages} {83} (\bibinfo {year} {1995})}\BibitemShut
  {NoStop}%
\bibitem [{\citenamefont {Eldar}\ \emph {et~al.}(2004)\citenamefont {Eldar},
  \citenamefont {Stojnic},\ and\ \citenamefont {Hassibi}}]{Eldar_sdp_2004}%
  \BibitemOpen
  \bibfield  {author} {\bibinfo {author} {\bibfnamefont {Y.~C.}\ \bibnamefont
  {Eldar}}, \bibinfo {author} {\bibfnamefont {M.}~\bibnamefont {Stojnic}}, \
  and\ \bibinfo {author} {\bibfnamefont {B.}~\bibnamefont {Hassibi}},\ }\href
  {\doibase 10.1103/PhysRevA.69.062318} {\bibfield  {journal} {\bibinfo
  {journal} {Phys. Rev. A}\ }\textbf {\bibinfo {volume} {69}},\ \bibinfo
  {pages} {062318} (\bibinfo {year} {2004})}\BibitemShut {NoStop}%
\bibitem [{\citenamefont {Chefles}(1998)}]{chefles_unambiguous_1998}%
  \BibitemOpen
  \bibfield  {author} {\bibinfo {author} {\bibfnamefont {A.}~\bibnamefont
  {Chefles}},\ }\href {\doibase 10.1016/S0375-9601(98)00064-4} {\bibfield
  {journal} {\bibinfo  {journal} {Phys. Lett. A}\ }\textbf {\bibinfo {volume}
  {239}},\ \bibinfo {pages} {339} (\bibinfo {year} {1998})}\BibitemShut
  {NoStop}%
\bibitem [{\citenamefont {Peres}\ and\ \citenamefont
  {Terno}(1998)}]{Peres_three_1998}%
  \BibitemOpen
  \bibfield  {author} {\bibinfo {author} {\bibfnamefont {A.}~\bibnamefont
  {Peres}}\ and\ \bibinfo {author} {\bibfnamefont {D.~R.}\ \bibnamefont
  {Terno}},\ }\href {\doibase 10.1088/0305-4470/31/34/013} {\bibfield
  {journal} {\bibinfo  {journal} {J. Phys. A}\ }\textbf {\bibinfo {volume}
  {31}},\ \bibinfo {pages} {7105} (\bibinfo {year} {1998})}\BibitemShut
  {NoStop}%
\bibitem [{\citenamefont {Chefles}\ and\ \citenamefont
  {Barnett}(1998)}]{chefles_optimum_1998}%
  \BibitemOpen
  \bibfield  {author} {\bibinfo {author} {\bibfnamefont {A.}~\bibnamefont
  {Chefles}}\ and\ \bibinfo {author} {\bibfnamefont {S.~M.}\ \bibnamefont
  {Barnett}},\ }\href {\doibase 10.1016/S0375-9601(98)00827-5} {\bibfield
  {journal} {\bibinfo  {journal} {Phys. Lett. A}\ }\textbf {\bibinfo {volume}
  {250}},\ \bibinfo {pages} {223} (\bibinfo {year} {1998})}\BibitemShut
  {NoStop}%
\bibitem [{\citenamefont {Croke}\ \emph {et~al.}(2006)\citenamefont {Croke},
  \citenamefont {Andersson}, \citenamefont {Barnett}, \citenamefont {Gilson},\
  and\ \citenamefont {Jeffers}}]{croke_maximum_2006}%
  \BibitemOpen
  \bibfield  {author} {\bibinfo {author} {\bibfnamefont {S.}~\bibnamefont
  {Croke}}, \bibinfo {author} {\bibfnamefont {E.}~\bibnamefont {Andersson}},
  \bibinfo {author} {\bibfnamefont {S.~M.}\ \bibnamefont {Barnett}}, \bibinfo
  {author} {\bibfnamefont {C.~R.}\ \bibnamefont {Gilson}}, \ and\ \bibinfo
  {author} {\bibfnamefont {J.}~\bibnamefont {Jeffers}},\ }\href
  {https://link.aps.org/doi/10.1103/PhysRevLett.96.070401} {\bibfield
  {journal} {\bibinfo  {journal} {Phys. Rev. Lett.}\ }\textbf {\bibinfo
  {volume} {96}} (\bibinfo {year} {2006})}\BibitemShut {NoStop}%
\bibitem [{\citenamefont {Jaynes}(2003)}]{ProbJaynes2003}%
  \BibitemOpen
  \bibfield  {author} {\bibinfo {author} {\bibfnamefont {E.~T.}\ \bibnamefont
  {Jaynes}},\ }\href {\doibase 10.1017/CBO9780511790423} {\emph {\bibinfo
  {title} {{Probability Theory}}}},\ edited by\ \bibinfo {editor}
  {\bibfnamefont {G.~L.}\ \bibnamefont {Bretthorst}}\ (\bibinfo  {publisher}
  {Cambridge University Press},\ \bibinfo {address} {Cambridge},\ \bibinfo
  {year} {2003})\BibitemShut {NoStop}%
\bibitem [{\citenamefont {Chaloner}\ and\ \citenamefont
  {Verdinelli}(1995)}]{95chaloner273}%
  \BibitemOpen
  \bibfield  {author} {\bibinfo {author} {\bibfnamefont {K.}~\bibnamefont
  {Chaloner}}\ and\ \bibinfo {author} {\bibfnamefont {I.}~\bibnamefont
  {Verdinelli}},\ }\href {\doibase 10.1214/ss/1177009939} {\bibfield  {journal}
  {\bibinfo  {journal} {Stat. Sci.}\ }\textbf {\bibinfo {volume} {10}},\
  \bibinfo {pages} {273} (\bibinfo {year} {1995})}\BibitemShut {NoStop}%
\bibitem [{\citenamefont {Hunter}(2003)}]{hunter_measurement_2003}%
  \BibitemOpen
  \bibfield  {author} {\bibinfo {author} {\bibfnamefont {K.}~\bibnamefont
  {Hunter}},\ }\href {https://link.aps.org/doi/10.1103/PhysRevA.68.012306}
  {\bibfield  {journal} {\bibinfo  {journal} {Phys. Rev. A}\ }\textbf {\bibinfo
  {volume} {68}} (\bibinfo {year} {2003})}\BibitemShut {NoStop}%
\end{thebibliography}%

\end{document}